\documentclass[12pt]{extarticle}

\usepackage{amsthm,amssymb,amsmath}
\usepackage{enumerate,graphicx,geometry}

\title{On the Probabilistic Degrees of Symmetric Boolean functions\thanks{A preliminary version of this paper will appear in the conference FSTTCS 2019.}}
\date{}
\author{Srikanth Srinivasan\thanks{Department of Mathematics, Indian Institute of Technology Bombay, Mumbai, India.  Email: \texttt{srikanth@math.iitb.ac.in}.  Supported by MATRICS grant MTR/2017/000958 awarded by SERB, Government of India.}\and
Utkarsh Tripathi\thanks{Department of Mathematics, Indian Institute of Technology Bombay, Mumbai, India.  Email: \texttt{utkarshtripathi.math@gmail.com}.  Supported by the Ph.D. Scholarship of NBHM, DAE, Government of India.}\and
S. Venkitesh\thanks{Department of Mathematics, Indian Institute of Technology Bombay, Mumbai, India.  Email: \texttt{venkitesh.mail@gmail.com}.  Supported by the Senior Research Fellowship of HRDG, CSIR, Government
of India.}}

\newtheorem{theorem}{Theorem}
\newtheorem{corollary}[theorem]{Corollary}
\newtheorem{lemma}[theorem]{Lemma}
\newtheorem{observation}[theorem]{Observation}

\newtheorem{definition}[theorem]{Definition}

\newtheorem{fact}[theorem]{Fact}
\newtheorem{remark}[theorem]{Remark}

\newcommand{\prob}[2]{\mathop{\mathrm{Pr}}_{#1}\left[#2\right]}
\newcommand{\avg}[2]{\mathop{\textbf{E}}_{#1}[#2]}
\newcommand{\poly}{\mathop{\mathrm{poly}}}

\newcommand{\F}{\mathbb{F}}

\newcommand{\AC}{\mathrm{AC}}

\newcommand{\mc}[1]{\mathcal{#1}}

\newcommand{\pdeg}{\mathrm{pdeg}}
\newcommand{\spec}{\mathrm{Spec}\,}
\newcommand{\per}{\mathrm{per}}
\newcommand{\sB}{s\mathcal{B}}
\newcommand{\Maj}{\mathrm{Maj}}
\newcommand{\Thr}{\mathrm{Thr}}
\newcommand{\MOD}{\mathrm{MOD}}
\newcommand{\OR}{\mathrm{OR}}
\newcommand{\AND}{\mathrm{AND}}
\newcommand{\EThr}{\mathrm{EThr}}

\newcommand{\mbf}{\mathbf}
\newcommand{\tsf}{\textsf}
\newcommand{\tbf}{\textbf}
\newcommand{\mb}{\mathbb}
\newcommand{\tx}{\text}

\newcommand{\Ex}{\tx{EX}}

\newcommand{\charac}{\tx{char}}

\newcommand{\vst}{\vspace{0.2cm}}
\newcommand{\vsf}{\vspace{0.5cm}}

\newcommand{\one}{\mbf{1}}
\newcommand{\supp}{\tx{\normalfont supp}}

\begin{document}
	
	\maketitle
	
	\begin{abstract}
		The probabilistic degree of a Boolean function $f:\{0,1\}^n\rightarrow \{0,1\}$ is defined to be the smallest $d$ such that there is a random polynomial $\mbf{P}$ of degree at most $d$ that agrees with $f$ at each point with high probability. Introduced by Razborov (1987), upper and lower bounds on probabilistic degrees of Boolean functions --- specifically symmetric Boolean functions --- have been used to prove explicit lower bounds, design pseudorandom generators, and devise algorithms for combinatorial problems. 
		
		In this paper, we characterize the probabilistic degrees of all symmetric Boolean functions up to polylogarithmic factors over all fields of fixed characteristic (positive or zero).
	\end{abstract}
	
	\section{Introduction}
	
	Studying the combinatorial and computational properties of Boolean functions by representing them using multivariate polynomials (over some field $\F$) is an oft-used technique in Theoretical Computer Science. Such investigations into the complexity of Boolean functions have led to many important advances in the area (see, e.g.~\cite{Beigel-survey, ODonnell, Williams-survey} for a large list of such results). 
	
	An ``obvious'' way of representing a Boolean function $f:\{0,1\}^n\rightarrow \{0,1\}$ is via a multilinear polynomial $P\in \F[x_1,\ldots,x_n]$ such that $P(a) = f(a)$ for all $a\in \{0,1\}^n$. While such a representation has the advantage of being unique, understanding the computational complexity of $f$ sometimes requires us to understand polynomial representations where we allow some notion of error in the representation. Many such representations have been studied, but we concentrate here on the notion of \emph{Probabilistic degree} of a Boolean function, introduced implicitly in a paper of Razborov~\cite{Razbo}. It is defined as follows.
	
	\begin{definition}[Probabilistic polynomial and Probabilistic degree]
		Given a Boolean function $f:\{0,1\}^n\rightarrow \{0,1\}$ and an $\varepsilon>0,$ an \emph{$\varepsilon$-error probabilistic polynomial} for $f$ is a random polynomial $\mbf{P}$ (with some distribution having finite support) over $\F[x_1,\ldots,x_n]$ such that for each $a\in \{0,1\}^n$,
		\[
		\prob{\mbf{P}}{\mbf{P}(a) \neq f(a)} \leq \varepsilon.
		\]
		We say that the degree of $\mbf{P}$, denoted $\deg(\mbf{P})$, is at most $d$ if the probability distribution defining $\mbf{P}$ is supported on polynomials of degree at most $d$. Finally, we define the \emph{$\varepsilon$-error probabilistic degree} of $f$, denoted $\pdeg^{\F}_\varepsilon(f)$, to be the least $d$ such that $f$ has an $\varepsilon$-error probabilistic polynomial of degree at most $d$.
		
		When the field $\F$ is clear from context, we use $\pdeg_\varepsilon(f)$ instead of $\pdeg^\F_\varepsilon(f).$
	\end{definition}
	
	Intuitively, if we think of multivariate polynomials as algorithms and degree as a notion of efficiency, then a low-degree probabilistic polynomial for a Boolean function $f$ is an efficient \emph{randomized} algorithm for $f$.
	
	The study of the probabilistic degree itself is by now a classical topic, and has had important repercussions for other problems. We list three such examples below, referring the reader to the papers for definitions and exact statements of the results. 
	
	\begin{itemize}
		\item Razborov~\cite{Razbo} showed strong upper bounds on the probabilistic degree of the OR function over fields of (fixed) positive characteristic. Along with lower bounds on the probabilistic degree of some symmetric Boolean functions,\footnote{Recall that a \emph{symmetric} Boolean function $f:\{0,1\}^n\rightarrow \{0,1\}$ is a function such that $f(x)$ depends only on the Hamming weight of $x$. Examples include the threshold functions, Parity (counting modulo $2$), etc..} this led to the first lower bounds for the Boolean circuit class $\AC^0[p]$, for prime $p$~\cite{Razbo, Smolensky87, Smolensky93}.
		\item Tarui~\cite{Tarui} and Beigel, Reingold and Spielman~\cite{BRS} showed upper bounds on the probabilistic degree of the OR function over any characteristic (and in particular over the reals). This leads to probabilistic degree upper bounds for the circuit class $\AC^0$, which was used by Braverman~\cite{Brav} to resolve a long-standing open problem of Linial and Nisan~\cite{LN} regarding pseudorandom generators for $\AC^0.$
		\item Alman and Williams~\cite{alman} showed that for constant error, the probabilistic degree of any symmetric Boolean function is at most $O(\sqrt{n})$, and used this to obtain the first subquadratic algorithm for an offline version of the Nearest Neighbour problem in the Hamming metric. 
	\end{itemize}
	
	In all the above results, it was important to understand the probabilistic degree of a certain class of symmetric Boolean functions. However, the problem of \emph{characterizing} the probabilistic degree of symmetric Boolean functions in general does not seem to have been considered. This is somewhat surprising, since this problem has been considered in a variety of other computational models, such as $\AC^0$ circuits of polynomial size~\cite{FKPS,BW}, $\AC^0[p]$ circuits of quasipolynomial size~\cite{Lu}, Approximate degree\footnote{A Boolean function $f:\{0,1\}^n\rightarrow \{0,1\}$ is said to have approximate degree at most $d$ if there is a degree $d$ polynomial $P\in \mathbb{R}[x_1,\ldots,x_n]$ such that at each $a\in \{0,1\}^n$, $|f(a)-P(a)|\leq 1/4.$}~\cite{paturi}  and constant-depth Perceptrons\footnote{These are constant-depth circuits that have an output Majority gate with $\AC^0$ circuits feeding into it.} of quasipolynomial size~\cite{ZhangBarringtonTarui}.
	
	\subparagraph*{Our result.} In this paper, we give an almost-complete understanding of the probabilistic degrees of all symmetric Boolean functions over all fields of fixed positive characteristic and characteristic $0$. For each Boolean function $f$ on $n$ variables, our upper bounds and lower bounds on $\pdeg(f)$ are separated only by polylogarithmic factors in $n$.
	
	We now introduce some notation and give a formal statement of our result.  We shall use the notation $[a,b]$ to denote an interval in $\mb{R}$ as well as an interval in $\mb{Z}$; the distinction will be clear from the context.  Throughout, fix some field $\F$ of characteristic $p$ which is either a fixed positive constant or $0$. Let $n$ be a growing integer parameter which will always be the number of input variables. We use $\sB_n$ to denote the set of all symmetric Boolean functions on $n$ variables. Note that each symmetric Boolean function $f:\{0,1\}^n\rightarrow \{0,1\}$ is uniquely specified by a string $\spec f:[0,n]\rightarrow \{0,1\}$, which we call the \emph{Spectrum} of $f$, in the sense that for any $a\in \{0,1\}^n$, we have
	\[
	f(a) = \spec f(|a|).
	\]
	
	Given a $f\in \sB_n$, we define the \emph{period of $f$}, denoted $\per(f),$ to be the smallest positive integer $b$ such that $\spec f(i) = \spec f(i+b)$ for all $i\in[0,n-b]$. We say $f$ is \emph{$k$-bounded} if $\spec f$ is constant on the interval $[k,n-k]$; let $B(f)$ denote the smallest $k$ such that $f$ is $k$-bounded.
	
	\subparagraph*{Standard decomposition of a symmetric Boolean function~\cite{Lu}.} Fix any $f\in \sB_n.$ Among all symmetric Boolean functions $f'\in \sB_n$ such that $\spec f'(i) = \spec f(i)$ for all $i\in[\lceil n/3\rceil,\lfloor 2n/3\rfloor],$ we choose a function $g$ such that $\per(g)$ is as small as possible. We call $g$ the \emph{periodic part} of $f$. Define $h\in \sB_n$ by $h = f\oplus g.$ We call $h$ the \emph{bounded part} of $f$. 
	
	We will refer to the pair $(g,h)$ as a \emph{standard decomposition} of the function $f$. Note that we have $f = g\oplus h.$
	
	\begin{observation}
		\label{obs:decomp}
		Let $f\in \sB_n$ and let $(g,h)$ be a standard decomposition of $f$. Then, $\per(g)\leq \lfloor n/3\rfloor$ and $B(h)\leq \lceil n/3 \rceil.$
	\end{observation}
	
	In this paper, we prove the following upper and lower bounds for the probabilistic degrees of symmetric Boolean functions. 
	
	\begin{theorem}[Upper bounds on probabilistic degree]
		\label{thm:main-ubd}
		Let $\F$ be a field of constant characteristic $p$ (possibly 0) and $n\in \mathbb{N}$ be a growing parameter. Let $f\in \sB_n$ be arbitrary and let $(g,h)$ be a standard decomposition of $f$. Then we have the following for any $\varepsilon > 0.$
		\begin{enumerate}
			\item If $\per(g) = 1,$ then $\pdeg^\F_\varepsilon(g) = 0.$ 
			
			If $\per(g)$ is a power of $p$, then $\pdeg^\F_\varepsilon(g) \leq \per(g),$ ~\cite{Lu}
			
			(Note that $\per(g)$ cannot be a power of $p$ if $p=0.$)
			\item $\pdeg^\F_\varepsilon(h) = \widetilde{O}(\sqrt{B(h)\log(1/\varepsilon)} + \log(1/\varepsilon))$ if $B(h) > 1$ and $0$ otherwise, and
			\item $\pdeg^\F_\varepsilon(f)=
			\left\{
			\begin{array}{ll}
			O(\sqrt{n\log(1/\varepsilon)}) & \text{if $\per(g) > 1$ and not a power of $p$,~\cite{alman}}\\
			O(\min\{\sqrt{n\log(1/\varepsilon)},\per(g)\}) & \text{if $\per(g)$ a power of $p$ and $B(h) =0$,}\\
			\widetilde{O}(\min\{\sqrt{n\log(1/\varepsilon)},\per(g) +  & \text{otherwise.}\\
			\ \ \sqrt{B(h)\log(1/\varepsilon)} + \log(1/\varepsilon) \}) & 	
			\end{array}
			\right.$
		\end{enumerate}
		where the $\widetilde{O}(\cdot)$ hides polylogarithmic factors in $n$ (and are independent of $\varepsilon$). When $p$ is positive, we can replaced the $\widetilde{O}(\cdot)$ with $O(\cdot)$ in all the above bounds.
	\end{theorem}
	
	We obtain almost (up to polylogarithmic factors) matching lower bounds for all symmetric Boolean functions over all fields and all errors.
	
	\begin{theorem}[Lower bounds on probabilistic degree]
		\label{thm:main-lbd}
		Let $\F$ be a field of constant characteristic $p$ (possibly $0$) and $n\in \mathbb{N}$ be a growing parameter. Let $f\in \sB_n$ be arbitrary and let $(g,h)$ be a standard decomposition of $f$. Then for any $\varepsilon\in [1/2^n,1/3]$, we have
		\begin{enumerate}
			\item $\pdeg_\varepsilon^\F(g) = \widetilde{\Omega}(\sqrt{n\log(1/\varepsilon)})$ if $\per(g) > 1$ and is not a power of $p$ and $\widetilde{\Omega}(\min\{\sqrt{n\log(1/\varepsilon)},\per(g)\})$ otherwise. 
			\item $\pdeg_\varepsilon^\F(h) = \widetilde{\Omega}(\sqrt{B(h)\log(1/\varepsilon)}+\log(1/\varepsilon))$ if $B(h) \geq 1,$ and
			\item $\pdeg_\varepsilon^\F(f) = 
			\left\{
			\begin{array}{ll}
			\widetilde{\Omega}(\sqrt{n\log(1/\varepsilon)}) & \text{if $\per(g) > 1$ and not a power of $p$,}\\
			\widetilde{\Omega}(\min\{\sqrt{n\log(1/\varepsilon)},\per(g)\}) & \text{if $\per(g)$ a power of $p$ and $B(h) =0$,}\\
			\widetilde{\Omega}(\min\{\sqrt{n\log(1/\varepsilon)},\per(g)  & \text{otherwise.}\\
			\ \ + \sqrt{B(h)\log(1/\varepsilon)} +\log(1/\varepsilon)\}) 
			\end{array}
			\right.$
		\end{enumerate}
		where the $\widetilde{\Omega}(\cdot)$ hides $\poly(\log n)$ factors (independent of $\varepsilon$).
	\end{theorem}
	
	\begin{remark}
		\label{rem:gap}
		A natural open question following our results is to remove the polylogarithmic factors separating our upper and lower bounds. We remark that in characteristic $0$, such gaps exist even for the very simple OR function despite much effort~\cite{MNV,HS,BHMS}. Over positive characteristic, there is no obvious barrier, but our techniques fall short of proving tight lower bounds for natural families of functions such as the Exact Threshold functions (defined in Section \ref{sec:prelims}).
	\end{remark}

	\subsection{Proof Outline}
	\label{sec:proof-outline}
	
	For the outline below, we assume that the field is of fixed positive characteristic $p$.
	
	\subparagraph*{Upper bounds.} Given a symmetric Boolean function $f$ on $n$ variables with standard decomposition $(g,h)$, it is easy to check that $\pdeg_\varepsilon(f) = O(\pdeg_\varepsilon(g) + \pdeg_\varepsilon(h)).$ So it suffices to upper bound the probabilistic degrees of periodic and bounded functions respectively.
	
	For periodic functions $g$ with period a power of $p$, Lu \cite{Lu} showed that the \emph{exact} degree of the Boolean functions is at most $\per(g)$. If the period is not a power of $p$, then we use the upper bound of Alman and Williams~\cite{alman} that holds for all symmetric Boolean functions (as we show below, this is nearly the best that is possible).
	
	For a $t$-constant function $h$ (defined in Section \ref{sec:ubd}), we use the observation that any $t$-constant function is essentially a linear combination of the threshold functions $\Thr_n^0,\ldots,\Thr_n^t$ (defined in Section~\ref{sec:prelims}) and so it suffices to construct probabilistic polynomials for $\Thr_n^i$ for $i\in [0,t]$.\footnote{We actually need to construct probabilistic polynomials for all the threshold functions simultaneously. We ignore this point in this high-level outline.}
	
	Our main technical upper bound is a new probabilistic degree upper bound of $O(\sqrt{t\log(1/\varepsilon)} + \log(1/\varepsilon))$ for any threshold function $\Thr_n^t.$ This upper bound interpolates smoothly between a classical upper bound of $O(\log(1/\varepsilon))$ due to Razborov~\cite{Razbo} for $t=1$ and a recent result of Alman and Williams~\cite{alman} that yields $O(\sqrt{n\log(1/\varepsilon)})$ for $t = \Omega(n)$.
	
	The proof of our upper bound is based on the beautiful inductive construction of Alman and Williams~\cite{alman} which gives their above-mentioned result. The key difference between our proof and the proof of~\cite{alman} is that we need to handle separately the case when the error $\varepsilon \leq 2^{-\Omega(t)}$.\footnote{This case comes up naturally in the inductive construction, even if one is ultimately only interested in the case when $\varepsilon$ is a constant.} In~\cite{alman}, this is a trivial case since any function on $n$ Boolean variables has an exact polynomial of degree $n$ which is at most $O(\sqrt{n\log(1/\varepsilon)})$ when $\varepsilon \leq 2^{-\Omega(n)}.$ In our setting, the correct bound in this case is $O(\log(1/\varepsilon))$, which is non-obvious. We obtain this bound by a suitable modification of Razborov's technique (for $t=1$) to handle larger thresholds. 
	
	\subparagraph*{Lower bounds.} Here, our proof follows a result of Lu~\cite{Lu}, who gave a characterization of symmetric Boolean functions that have quasipolynomial-sized $\AC^0[p]$ circuits.\footnote{Recall that an $\AC^0[p]$ circuit is a constant-depth circuit made up of gates that can compute the Boolean functions AND, OR, NOT and $\MOD_p$ (defined below).} To show circuit lower bounds for a symmetric Boolean function $h$, Lu showed how to convert a circuit $C$ computing $h$ to a circuit $C'$ computing either the Majority or a $\MOD_q$ function (where $q$ and $p$ are relatively prime). Since both of these are known to be hard for $\AC^0[p]$~\cite{Razbo,Smolensky87}, we get the lower bound. 
	
	Lu's basic idea was to use a few restrictions\footnote{A restriction of a Boolean function is obtained by setting some of its input variables to constants in $\{0,1\}.$} of $h$ along with some additional circuitry to compute either Majority or $\MOD_q.$ These functions are also known to have large probabilistic degree (in fact, this is the source of the $\AC^0[p]$ lower bound), and so this high-level idea seems applicable to our setting as well. Indeed we do use this strategy, but our proofs are different when it comes down to the details. As Lu's aim was to derive optimal circuit lower bounds for $h$, his reductions were tailored towards using as small an amount of additional circuitry as possible. Our focus, however, is to prove the best possible probabilistic degree lower bound, so we would like our reductions to be computable by polynomials of small degree. This makes the actual reductions quite different.\footnote{In an earlier version of this paper, we actually  used Lu's reductions (and variants thereof) directly in the setting of probabilistic polynomials. This still works in certain parameter regimes because the additional circuitry itself has low \emph{probabilistic} degree. However, in the setting of small error, this strategy seems to yield suboptimal results.}
%
	
	\section{Preliminaries}
	\label{sec:prelims}
	
	\subparagraph*{Some Boolean functions.} 
	Fix some positive $n\in \mathbb{N}$. The \emph{Majority} function $\Maj_n$ on $n$ Boolean variables accepts exactly the inputs of Hamming weight greater than $ n/2.$ For $t\in [0,n]$, the \emph{Threshold} function $\Thr^t_n$ accepts exactly the inputs of Hamming weight at least $t$; and similarly,  the \emph{Exact Threshold} function $\EThr^t_n$ accepts exactly the inputs of Hamming weight exactly $t$. Finally, for $b\in [2,n]$ and $i\in [0,b-1]$, the function $\MOD^{b,i}_n$ accepts exactly those inputs $a$ such that $|a| \equiv i\pmod{b}.$ In the special case that $i=0$, we also use $\MOD^b_n.$
	
	\begin{fact}
		\label{fac:pdeg}
		We have the following simple facts about probabilistic degrees. Let $\F$ be any field.
		\begin{enumerate}
			\item (Error reduction~\cite{HS}) For any $\delta < \varepsilon \leq 1/3$ and any Boolean function $f$, if $\mbf{P}$ is an $\varepsilon$-error probabilistic polynomial for $f$, then $\mbf{Q} = M(\mbf{P}_1,\ldots,\mbf{P}_\ell)$ is a $\delta$-error probabilistic polynomial for $f$ where $M$ is the exact multilinear polynomial for $\Maj_\ell$ and $\mbf{P}_1,\ldots,\mbf{P}_\ell$ are independent copies of $\mbf{P}.$ In particular, we have $\pdeg_{\delta}^\F(f) \leq \pdeg_{\varepsilon}^\F(f)\cdot O(\log(1/\delta)/\log(1/\varepsilon)).$ 
			\item (Composition) For any Boolean function $f$ on $k$ variables and any Boolean functions $g_1,\ldots,g_k$ on a common set of $m$ variables,  let $h$ denote the natural composed function $f(g_1,\ldots,g_k)$ on $m$ variables. Then, for any $\varepsilon, \delta > 0,$ we have $\pdeg_{\varepsilon + k\delta}^\F(h) \leq \pdeg_\varepsilon^\F(f)\cdot \max_{i\in [k]} \pdeg_\delta^\F(g_i).$
			\item (Sum) Assume that $f,g_1,\ldots,g_k$ are all Boolean functions on a common set of $m$ variables such that $f = \sum_{i\in [k]}g_i$. Then, for any $\delta > 0,$ we have $\pdeg_{k\delta}^\F(f) \leq \max_{i\in [k]} \pdeg_\delta^\F(g_i).$
		\end{enumerate}
	\end{fact}

	\subsection{Some previous results on probabilistic degree}

	The following upper bounds on probabilistic degrees of OR and AND functions were proved by Razborov~\cite{Razbo} and Smolensky~\cite{Smolensky87} in the case of positive characteristic and Tarui~\cite{Tarui} and Beigel, Reingold and Spielman~\cite{BRS} in the general case. For the latter, we state a slightly tighter result that follows from \cite[Lemma 8]{Brav}.
	
	\begin{lemma}[Razborov's upper bound on probabilistic degrees of OR and AND]
		\label{lem:razb}
		Let $\F$ be a field of characteristic $p$. For $p > 0$, we have
		\begin{equation}
		\label{eq:razb-p>0}
		\pdeg^\F_\varepsilon(\OR_n) = \pdeg^\F_\varepsilon(\AND_n) \leq p\lceil \log(1/\varepsilon))\rceil.
		\end{equation}
		For any $p$, we have
		\begin{equation}
		\label{eq:razb-p=0}
		\pdeg^\F_\varepsilon(\OR_n) = \pdeg^\F_\varepsilon(\AND_n) \leq  4\lceil \log n\rceil \cdot \lceil \log(1/\varepsilon)\rceil.
		\end{equation}
		Further, the probabilistic polynomials have one-sided error in the sense that on the all $0$ input, they output $0$ with probability $1$.
	\end{lemma}
	
	We now recall two probabilistic degree lower bounds due to Smolensky~\cite{Smolen,Smolen93}, building on the work of Razborov~\cite{Razbo}.
	
	\begin{lemma}[Smolensky's lower bound for close-to-Majority functions]
		\label{lem:majlbd}
		For any field $\F$, any $\varepsilon \in (1/2^n, 1/5),$ and any Boolean function $g$ on $n$ variables that agrees with $\Maj_n$ on a $1-\varepsilon$ fraction of its inputs, we have
		\[\pdeg_{\varepsilon}^\F(g) = \Omega(\sqrt{n\log(1/\varepsilon)}).\]
	\end{lemma}
	
	\begin{lemma}[Smolensky's lower bound for MOD functions]
		\label{lem:modlbd}
		For $2\leq b\leq n/2$, any $\F$ such that $\charac(\F) = p$ is coprime to $b$, any $\varepsilon \in (1/2^n, 1/(3b))$, there exists an $i\in[0,b-1]$ such that
		\[\pdeg_{\varepsilon}^\F(\MOD^{b,i}_n) = \Omega(\sqrt{n\log(1/b\varepsilon)}).\]
	\end{lemma}
	
	\begin{remark}
		\label{rem:smolenskymod}
		From the above lemma, it also easily follows that if $b\leq n/4$, then for \emph{every} $i\in[0,b-1]$, we have $\pdeg_{\varepsilon}^\F(\MOD^{b,i}_n) = \Omega(\sqrt{n\log(1/b\varepsilon)}).$
		This is the usual form in which Smolensky's lower bound is stated. The above form is slightly more useful to us because it holds for $b$ up to $n/2$.
	\end{remark}
	
	We will also need the following result of Alman and Williams~\cite{alman}.
	
	\begin{lemma}
		\label{lem:AW}
		Let $\F$ be any field. For any $n\geq 1,\varepsilon > 0$ and $f\in \sB_n,$ $\pdeg_\varepsilon^\F(f) = O(\sqrt{n\log(1/\varepsilon)}).$
	\end{lemma}

	\subsection{A string lemma}
	
	Given a function $w:I\rightarrow \{0,1\}$ where $I\subseteq \mathbb{N}$ is an interval, we think of $w$ as a string from the set $\{0,1\}^{|I|}$ in the natural way. For an interval $J\subseteq I,$ we denote by $w|_{J}$ the substring of $w$ obtained by restriction to $J$.

	The following simple lemma can be found, e.g. as a consequence of~\cite[Chapter I, Section 2, Theorem 1]{johnson-stringlemma}. 
	\begin{lemma}
		\label{lem:string}
		Let $w\in \{0,1\}^+$ be any non-empty string and $u,v\in \{0,1\}^+$ such that $w = uv = vu$. Then there exists a string $z\in \{0,1\}^+$ such that $w$ is a power of $z$ (i.e. $w= z^k$ for some $k\geq 2$).
	\end{lemma}
	
	\begin{corollary}
		\label{cor:string}
		Let $g\in \sB_n$ be arbitrary with $\per(g) = b > 1.$ Then for all $i,j\in [0,n-b+1]$ such that $i \not\equiv j \pmod{b}$, we have $\spec g|_{[i,i+b-1]} \neq \spec g|_{[j,j+b-1]}.$
	\end{corollary}
	
	\begin{proof}
		Suppose $\spec g|_{[i,i+b-1]} = \spec g|_{[j,j+b-1]}$ for some $i \not\equiv j \pmod{b}.$ Assume without loss of generality that $i < j < i+b.$ Let $u = \spec g|_{[i,j-1]}, v = \spec g|_{[j,i+b-1]}, w = \spec g|_{[i+b,j+b-1]}$. Then $u=w$ and the assumption $uv=vw$ implies $uv=vu$. By Lemma~\ref{lem:string}, there exists a string $z$ such that $uv = z^k$ for $k\geq 2$ and therefore $\per(g) < b$. This contradicts our assumption on $b$.
	\end{proof}

	\section{Upper bounds}\label{sec:ubd}
	
	In this section, we will first prove upper bounds on the probabilistic degree of a special class of symmetric Boolean functions that we call \emph{$t$-constant functions}, and then use it to prove Theorem \ref{thm:main-ubd}.
	
	\subsection{Upper bound on probabilistic degree of $t$-constant functions}
	
	\begin{definition}[$t$-constant function]
		For any positive $n\in\mb{N}$ and $t\in[0,n]$, a Boolean function $f\in \sB_n$ is said to be \emph{$t$-constant} if $f|_{\{x:|x|\ge t\}}$ is a constant, that is, $\spec f|_{[t,n]}$ is a constant.
	\end{definition}
	The following observation is immediate.
	\begin{observation}\label{obs:t-bdd}
		A Boolean function $f:\{0,1\}^n\to\{0,1\}$ is $t$-constant if and only if $f=\sum_{j=0}^ta_j\Thr_n^j$, for some $a_0,\ldots,a_t\in\{-1,0,1\}$.  In other words, $f$ is $t$-constant if and only if there exists a linear polynomial $g(Y_0,\ldots,Y_t)=a_0Y_0+\cdots+a_tY_t\in\F[Y_0,\ldots,Y_t]$ with $a_j\in\{-1,0,1\},\,j\in[0,t]$ such that $f=g(\Thr_n^0,\ldots,\Thr_n^t)$.
	\end{observation}
	
	We will prove an upper bound on the probabilistic degree of $t$-constant Boolean functions.  For this, we first generalize the notion of probabilistic polynomial and probabilistic degree to a \emph{tuple} of Boolean functions.  This generalization was implicit in \cite{alman}.
	
	\begin{definition}[Probabilistic poly-tuple and probabilistic degree]
		Let $f=(f_1,\ldots,f_m):\{0,1\}^n\to\{0,1\}^m$ be an $m$-tuple of Boolean functions and $\varepsilon\in(0,1)$.  An $\varepsilon$-error probabilistic poly-tuple for $f$ is a random $m$-tuple of polynomials $\mbf{P}$ (with some distribution having finite support) from $\F[X_1,\ldots,X_n]^m$ such that
		\[
		\prob{\mbf{P}}{\mbf{P}(x)\ne f(x)}\le\varepsilon,\quad\tx{for all }x\in\{0,1\}^n.
		\]
		We say that the degree of $\mbf{P}$ is at most $d$ if $\mbf{P}$ is supported on $m$-tuples of polynomials $P=(P_1,\ldots,P_m)$ where each $P_i$ has degree at most $d$.  Finally we define the $\varepsilon$-error probabilistic degree of $f$, denoted by $\pdeg_\varepsilon^\F(f)$, to be the least $d$ such that $f$ has an $\varepsilon$-error probabilistic poly-tuple of degree at most $d$.	
	\end{definition}
	
	We make a definition for convenience.
	
	\begin{definition}[Threshold tuple]
		For positive $n\in\mb{N},\,t\in[0,n]$, an $(n,t)$-threshold tuple is any tuple of Boolean functions $(\Thr_n^{t_1},\ldots,\Thr_n^{t_m})$, with $t_1,\ldots,t_m\in[0,t]$ and $\max\{t_1,\ldots,t_m\}\leq t$.
	\end{definition}
	The main theorem of this subsection is the following.
	
	\begin{theorem}\label{thm:pdeg-thr-vec}
		For any positive $n\in\mb{N},\,t\in[0,n]$, if $T$ is an $(n,t)$-threshold tuple and $\varepsilon\in(0,1/3)$, then
		\[
		\pdeg_\varepsilon(T)=\begin{cases}
		\widetilde O(\sqrt{t\log(1/\varepsilon)}+\log(1/\varepsilon)),&\charac(\F)=0,\\
		O(\sqrt{t\log(1/\varepsilon)}+\log(1/\varepsilon)),&\charac(\F)=p>0.
		\end{cases}.
		\]
	\end{theorem}
	
	As a corollary to the above theorem, we get an upper bound for the probabilistic degree of $t$-constant functions.
	
	\begin{corollary}\label{cor:pdeg-t-bdd}
		For any $t$-constant Boolean function $f:\{0,1\}^n\to\{0,1\}$ and $\varepsilon\in(0,1/3)$,
		\[
		\pdeg_\varepsilon(f)=\begin{cases}
		\widetilde O(\sqrt{t\log(1/\varepsilon)}+\log(1/\varepsilon)),&\charac(\F)=0,\\
		O(\sqrt{t\log(1/\varepsilon)}+\log(1/\varepsilon)),&\charac(\F)=p>0.
		\end{cases}.
		\]
	\end{corollary}
	\begin{proof}
		By Observation \ref{obs:t-bdd}, there exists $g(Y_0,\ldots,Y_t)=a_0Y_0+\cdots+a_tY_t\in\F[Y_0,\ldots,Y_t]$ with $a_j\in\{-1,0,1\},\,j\in[0,t]$ such that $f=g(\Thr_n^0,\ldots,\Thr_n^t)$.  We note that $\deg g=1$.  So by Theorem \ref{thm:pdeg-thr-vec}, we get
		\[
		\pdeg_\varepsilon(f)=\deg g\cdot\pdeg_\varepsilon(\Thr_n^0,\ldots,\Thr_n^t)=\begin{cases}
		\widetilde O(\sqrt{t\log(1/\varepsilon)}+\log(1/\varepsilon)),&\charac(\F)=0,\\
		O(\sqrt{t\log(1/\varepsilon)}+\log(1/\varepsilon)),&\charac(\F)=p>0.
		\end{cases}.\qedhere
		\]
	\end{proof}
	
	\paragraph*{High-level outline of the proof.} The basic strategy behind the inductive construction of probabilistic poly-tuples is due to Alman and Williams~\cite{alman}. We describe the construction of an $\varepsilon$-error probabilistic polynomial for a single threshold $\Thr_n^t$ (the construction for a tuple is similar). Assume that, by induction, we already have probabilistic polynomials $\mbf{T}_{m,t,\varepsilon}$ for $\Thr^t_m$ where $m < n$. The idea is to try  to use $\mbf{T}_{m,t,\varepsilon}$ for $m < n$ to compute $\Thr_n^t(x).$ We do this by sampling: we sample a random subvector $\hat{\mbf{x}}$ of length $n/10$ of $x$ by sampling uniform random entries of $x$ with replacement. If the Hamming weight of $x$ is ``sufficiently far'' from the threshold $t$, then the weight of $\hat{\mbf{x}}$ is on the ``same side'' of $t/10$ as $x$ is of $t$ w.h.p. (say at least $1-\varepsilon/4$); in particular, in this case $\mbf{P}_{n/10,t/10,\varepsilon/4}$ gives the right answer with probability $1-\varepsilon/4$ and we are done. However, if $|x|$ is ``not sufficiently far'' from $t$, then we need to do something else: here, we simply interpolate a polynomial that outputs the right answer on these values (see Theorem~\ref{thm:aw-inter} below). Finally, to check which of ``far'' or the ``not far'' cases we are in, we again use the inductive hypothesis on the subvector $\hat{\mbf{x}},$ which again gives the right answer with probability $1-\varepsilon/4$. Putting these things together yields the $\varepsilon$-error probabilistic polynomial.
	
	In the analysis of the construction above, the distance parameter (say $\theta$) that determines ``far'' vs. ``not far'' comes from the concentration properties of Bernoulli random variables (see Lemma~\ref{lem:chernoff-absolute} below). in our setting, $\theta$ is roughly $\sqrt{t\log(1/\varepsilon)})$. In particular, to check that $|x|$ is not much larger than $t$, we need to apply a probabilistic polynomial for the threshold function $\Thr_{n/10}^{t/10 + \theta}$ to the random vector $\hat{\mbf{x}}$. Here, to keep the threshold parameter bounded by $t$, we need that $\log(1/\varepsilon)$ is not much larger than $t$, or equivalently that $\varepsilon$ is not much smaller than $2^{-t}.$ 
	
	When $\varepsilon$ does fall below $2^{-t},$\footnote{Note that this can occur even if we are only interested in the case of (say) constant error. Since the inductive strategy causes the error to drop at each stage, even if we start with constant $\varepsilon,$ after a few stages we end up in the setting where $\varepsilon < 2^{-t}.$} we need to do something different, as the above inductive strategy fails. This case does not occur in~\cite{alman} since there $t = \Omega(n)$ and when $\varepsilon \leq 2^{-n},$ we can always use an exact polynomial representation of the threshold function (which has degree $n = O(\sqrt{n\log(1/\varepsilon)})$). In our setting, though, we aim for a bound of $\tilde{O}(\log(1/\varepsilon))$ in this case, which is non-trivial. To handle this case, we use a different construction, which is a modification of Razborov's probabilistic polynomial construction for the $\OR$ function (Lemma~\ref{lem:razb} above). This changes the base case of the induction and certain elements of the inductive analysis. Overall, though, we are able to use these ideas to obtain a probabilistic polynomial of degree $\tilde{O}(\sqrt{t\log(1/\varepsilon)} + \log(1/\varepsilon))$ (and only a constant-factor loss when the characteristic $p > 0$).
	
	\subsubsection{Proof of Theorem~\ref{thm:pdeg-thr-vec}}
	
	Before we prove Theorem \ref{thm:pdeg-thr-vec}, we will gather a few results that we require.  The following lemma is a particular case of Bernstein's inequality (Theorem 1.4, \cite{dubhashi_panconesi_2009}).
	\begin{lemma}\label{lem:chernoff-absolute}
		Let $X_1,\ldots,X_m$ be independent and identically distributed Bernoulli random variables with mean $q$.  Let $X=\sum_{i=1}^mX_i$.  Then for any $\theta>0$,
		\[
		\prob{}{|X-mq|>\theta}\le2\exp\bigg(-\frac{\theta^2}{2mq(1-q)+2\theta/3}\bigg).
		\]
	\end{lemma}
	
	We will also need the following polynomial construction.
	\begin{theorem}[Lemma 3.1, \cite{alman}]
	\label{thm:aw-inter}
		For any symmetric Boolean function $f:\{0,1\}^n\to\{0,1\}$ and integer interval $[a,b]\subseteq[0,n]$, there exists a symmetric multilinear polynomial $\Ex_{[a,b]}f\in\mb{Z}[X_1,\ldots,X_n]$ such that $\deg(\Ex_{[a,b]}f)\le b-a$ and $\spec(\Ex_{[a,b]}f)|_{[a,b]}=\spec f|_{[a,b]}$.
	\end{theorem}

	\begin{remark}
	\label{rem:awpoly}
	In particular, the polynomial $\Ex_{[a,b]}f$ may be interpreted as a polynomial over any field $\F$ satisfying the above property.
	\end{remark}

	We will now prove Theorem \ref{thm:pdeg-thr-vec}.
	\begin{proof}[\normalfont\tbf{Proof of Theorem \ref{thm:pdeg-thr-vec}}]
		For any $a=(a_1,\ldots,a_k),b=(b_1,\ldots,b_k)\in\F^k$, fix the notation $a*b=(a_1b_1,\ldots,a_kb_k)$.  Throughout, the notation $\one$ will denote the constant-1 vector of appropriate length.
		
		For positive characteristic $p$, we prove that for any positive $n\in\mb{N},\,t\in[0,n]$ and $\varepsilon\in(0,2^{-100})$, any $(n,t)$-threshold tuple $T$ has an $\varepsilon$-error probabilistic poly-tuple $\mbf{T}$ of degree at most $A_p\sqrt{t\log(1/\varepsilon)} + B_p\log(1/\varepsilon)$, for constants $A_p=B_p=6,400,000p$ (we make no effort to optimize the constants).    For $p=0$, we prove a similar result with a degree bound of $A_0\log n\cdot \sqrt{t\log(1/\varepsilon)} + B_0\log n \cdot \log(1/\varepsilon)$, for $A_0=B_0=64,000,000$.  This will prove the theorem for $\varepsilon \le 2^{-100}.$ To prove the theorem for all $\varepsilon\leq 1/3,$ we use error reduction (Fact~\ref{fac:pdeg}) and reduce the error to $2^{-100}$ and then apply the result for small error.
		
		The proof is by induction on the parameters $n,t$ and $\varepsilon.$  At any stage of the induction, given an \((n,t)\)-threshold tuple with error parameter \(\varepsilon\), we construct the required probabilistic poly-tuple by using the probabilistic poly-tuples (guaranteed by inductive hypothesis) for suitable threshold poly-tuples with \(n/10\) inputs and error parameter \(\varepsilon/4\).  Thus the base cases of the induction are as follows.
		
		\vst
		\tsf{Base Case:}\quad  Suppose $n\le 10$.  Let $T=(T_1,\ldots,T_m)$ be an $(n,t)$-threshold tuple.  Let $Q_1,\ldots,Q_m$ be the unique multilinear polynomial representations of $T_1,\ldots,T_m$ respectively.  Then $Q=(Q_1,\ldots,Q_m)$ is an $\varepsilon$-error probabilistic poly-tuple for $T$, for all $\varepsilon\in(0,2^{-100})$, with $\deg Q\le n=10\leq B_p\log(1/\varepsilon)$ for all $p$ (including $0$). Hence the claim is proved in this case.
		
		\vst
		\tsf{Base Case:}\quad  Suppose $\varepsilon\le2^{-t/160000}$.  Let $T=(T_1,\ldots,T_m)=(\Thr_n^{t_1},\ldots,\Thr_n^{t_m})$ be any $(n,t)$-threshold tuple and let $r=6400000\log(1/\varepsilon)$.
		
		Suppose $n\le r$.  Let $Q_1,\ldots,Q_m$ be the unique multilinear representations of $T_1,\ldots,T_m$ respectively.  Then $Q=(Q_1,\ldots,Q_m)$ is an $\varepsilon$-error probabilistic polynomial with $\deg Q\le n\le r$.  This proves the claim in this case.
		
		Now suppose $n>r$.  We first describe how to construct the probabilistic poly-tuple $\mbf{P}$ in this case. Assume for now that $\charac(\F) = p > 0.$
		
		\begin{itemize}
		\item Let $P_1=(\Ex_{[0,r]}T_1,\ldots,\Ex_{[0,r]}T_m)$.  Then $\deg P_1\le r$.
		\item   Choose a uniformly random hash function $\mbf{H}:[n]\to[r]$ and let $\mbf{S}_j=\mbf{H}^{-1}(j),\,j\in[r]$. Choose $\alpha_i\sim\F_p,\,i\in[n]$ independently and uniformly at random and define $\mbf{L}_j(x)=\sum_{i\in \mbf{S}_j}\alpha_ix_i,\,x\in\{0,1\}^n,\,j\in[r]$. For $i\in[m]$, let $\mbf{P}_2^{(i)}=Q_r^{(i)}(\mbf{L}_1^{p-1},\ldots,\mbf{L}_r^{p-1})$, where $Q_r^{(i)}$ is the unique multilinear polynomial representation of $\Thr_r^{t_i}$.  Let $\mbf{P}_2=(\mbf{P}_2^{(1)},\ldots,\mbf{P}_2^{(m)})$.  Note that $\deg(\mbf{P}_2)\leq (p-1)\cdot \left(\max_i \deg(Q_r^{(i)})\right)\leq (p-1)\cdot r$. 
		\item Define $\mbf{P}=\one-(\one-P_1)*(\one-\mbf{P}_2)$, that is, $\mbf{P}=(\mbf{P}^{(1)},\ldots,\mbf{P}^{(m)})$, where $\mbf{P}^{(i)}=\OR_2(P_1^{(i)},\mbf{P}_2^{(i)})$, for all $i\in[m]$. We have $\deg(\mbf{P})\leq \deg(P_1)+\deg(\mbf{P}_2)\leq p\cdot r\leq B_p\log(1/\varepsilon).$
		\end{itemize}

%
		
		We now show that $\mbf{P}$ is indeed an $\varepsilon$-error probabilistic poly-tuple for $T$. Note that since $\varepsilon\le 2^{-t/160000}$, we have $r=6400000\log(1/\varepsilon)\ge 40t>t$.    Thus $t_i\le t\le r$, for all $i\in[m]$.  Now fix any $a\in\{0,1\}^n$.  Let $Z_a=\{i\in[m]:\Thr_n^{t_i}(a)=0\}$ and $N_a=\{i\in[m]:\Thr_n^{t_i}(a)=1\}$.  So we have $|a|<t_i\le t\le r$ and hence $\Ex_{[0,r]}T_i(a)=0$, for all $i\in Z_a$.  Also $|(\mbf{L}_1^{p-1}(a),\ldots,\mbf{L}_r^{p-1}(a))|\le|a|<t_i$ w.p.1, and so $\mbf{P}_2^{(i)}(a)=Q_r^{(i)}((\mbf{L}_1^{p-1}(a),\ldots,\mbf{L}_r^{p-1}(a)))=0$ w.p.1,  for all $i\in Z_a$ simultaneously.  Thus $\mbf{P}^{(i)}(a)=0$ w.p.1, for all $i\in Z_a$ simultaneously.
		
		Further we have $|a|\ge t_i$, for all $i\in N_a$.  We will now show that $\mbf{P}^{(i)}(a)=1$ w.p. at least $1-\varepsilon$, for all $i\in N_a$ simultaneously.  If $|a|\le r$, then again $P_1^{(i)}(a)=1$, for all $i\in N_a$ and so $\mbf{P}^{(i)}(a)=1$ w.p.1.  Now suppose $|a|\ge r$ (note that in this case, $N_a = [m]$).  Without loss of generality, assume $t_1\le\cdots\le t_m=t$.  Then we have $\mbf{P}_2^{(1)}(a)\ge\cdots\ge\mbf{P}_2^{(m)}(a)$ w.p. 1, under the order $1>0$.  So it is enough to show that $\mbf{P}_2^{(m)}(a)=1$ w.p. at least $1-\varepsilon$.
		
		Define $I(\mbf{H})=\{j\in[r]:\supp(a)\cap\mbf{S}_j\ne\emptyset\}$.  We get
		\begin{align*}
		\prob{}{\mbf{P}_2^{(m)}(a)=0}&=\prob{}{\mbf{P}_2^{(m)}(a)=0\,\,\Big\vert\,\,|I(\mbf{H})|<r/10}\cdot\prob{}{|I(\mbf{H})|<r/10}\\
		&\quad+\prob{}{\mbf{P}_2^{(m)}(a)=0\,\,\Big\vert\,\,|I(\mbf{H})|\ge r/10}\cdot\prob{}{|I(\mbf{H})|\ge r/10}\\
		&\le\prob{}{|I(\mbf{H})|<r/10}+\max_{\mbf{H}:|I(\mbf{H})|\ge r/10}\prob{}{\mbf{P}_2^{(m)}(a)=0\,\,\Big\vert\,\,\mbf{H}}.
		\end{align*}
		Since the function $\mbf{H}:[n]\rightarrow [r]$ is chosen uniformly at random, the probability that $I(\mbf{H}) \subseteq I$ for any set $I\subseteq [r]$ is $(|I|/r)^{|a|}.$ Using the fact that $|a| > r$ and the union bound, we get
		\[
		\prob{}{|I(\mbf{H})|<r/10}\le\sum_{I\subseteq[r],\,|I|=r/10}\prob{}{I(\mbf{H})\subset I}\le{r\choose r/10}\frac{1}{10^r}\le\frac{1}{4^r}\le\frac{\varepsilon}{4}.
		\]
		Now fix any $\mbf{H}$ such that $|I(\mbf{H})|\ge r/10$, and let $\ell = |I(\mbf{H})|$. Note that $\mbf{P}_2^{(m)}(a)$ is $0$ if and only if at most $t-1$ many $\mbf{L}_j(a)$ are non-zero. We consider only $j\in I(\mbf{H})$. For each $j\in I(\mbf{H})$, the probability that $\mbf{L}_j(a)$ is non-zero is $1-1/p \geq 1/2.$ Let $\mbf{Z}$ denote the number of such $\mbf{L}_j$ ($j\in I(\mbf{H})$). Thus, the expected value of $\mbf{Z}$ is at least $\ell/2 \geq r/20\geq 2t$. Thus, by Lemma \ref{lem:chernoff-absolute}, 
		\begin{align*}
		\prob{}{\mbf{P}_2^{(m)}(a)=0\,\Big\vert\,\mbf{H}}&=\prob{}{|I(\mbf{H})\cap\{j:\mbf{L}_j(a)=1\}|\le t-1\ |\ \mbf{H}}\\
		&\leq  \prob{}{|\mbf{Z} - \avg{}{\mbf{Z}}| \geq \ell/4\ |\ \mbf{H}}\\
		&\le2\exp\bigg(-\frac{\ell^2/16}{2\cdot \ell \cdot (1/4) + (2/3)\cdot (\ell/4)}\bigg)<\frac{\varepsilon}{2}.
		\end{align*}
		where for the final inequality we have used the fact that $\ell \geq r/10\geq 640000\log(1/\varepsilon) $.   Thus $\prob{}{\mbf{P}_2^{(m)}(a)=0}\le\varepsilon$, proving the claim when $\charac(\F)=p>0$.
		
		\vst
		Now suppose $\charac(\F)=0$.  Then we use the same construction as above except for one change: for $i\in[m]$ we let $\mbf{P}_2^{(i)}=Q_r^{(i)}(\mbf{O}_1,\ldots,\mbf{O}_r)$, where $Q_r^{(i)}$ is the unique multilinear polynomial representation of $\Thr_r^{t_i}$, and for $j\in[r]$, $\mbf{O}_j$ is the $1/4$-error probabilistic polynomial for $\OR_{\mbf{S}_j}$, the $\OR$ function on variables $(X_k:k\in\mbf{S}_j)$, given to us by Lemma~\ref{lem:razb}.  One can verify that the degree in this case is bounded as above by $10r\log n\leq B_0 \log n\cdot \log(1/\varepsilon).$ The rest of the analysis follows similarly, proving the base case when $\charac(\F)=0$.
		
		\vsf
		\tsf{Inductive Construction.}\quad  For any positive characteristic $p$, any $n'<n,\,t'\in[0,n']$ and $\varepsilon'\in(0,2^{-100})$, assume the existence of an $\varepsilon'$-error probabilistic poly-tuple for any $(n',t')$-threshold tuple, with degree at most $A_p\sqrt{t'\log(1/\varepsilon')}+B_p\log(1/\varepsilon')$; similarly, for characteristic zero, assume we have a probabilistic poly-tuple of degree $A_0\log n\cdot \sqrt{t'\log(1/\varepsilon')} + B_0\log n \cdot \log(1/\varepsilon')$.  
		
		We now consider an $(n,t)$-threshold tuple $T=(T_1,\ldots,T_m)=(\Thr_n^{t_1},\ldots,\Thr_n^{t_m})$.  Assume that the parameter $\varepsilon > 2^{-t/160000}$ since otherwise we can use the construction from the base case. Define
		\begin{align*}
		T'&=(T_1',\ldots,T_m')=\Big(\Thr_{n/10}^{t_1/10},\ldots,\Thr_{n/10}^{t_m/10}\Big),\\
		T_+''&=(T_{1,+}'',\ldots,T_{m,+}'')=\Big(\Thr_{n/10}^{t_1/10+20\sqrt{t\log(1/\varepsilon)}},\ldots,\Thr_{n/10}^{t_m/10+20\sqrt{t\log(1/\varepsilon)}}\Big),\\
		T_-''&=(T_{1,-}'',\ldots,T_{m,-}'')=\Big(\Thr_{n/10}^{t_1/10-20\sqrt{t\log(1/\varepsilon)}},\ldots,\Thr_{n/10}^{t_m/10-20\sqrt{t\log(1/\varepsilon)}}\Big).
		\end{align*}
		By induction hypothesis, let $\mbf{T}',\mbf{T}_+'',\mbf{T}_-''$ be $\varepsilon/4$-error probabilistic poly-tuples for $T',T_+'',T_-''$ respectively.  Let $\mbf{N}''=(\one-\mbf{T}_+'')*\mbf{T}_-''$.  For any $x\in\{0,1\}^n$, choose a random subvector $\hat{\mbf{x}}\in\{0,1\}^{n/10}$ with each coordinate of $\hat{\mbf{x}}$ chosen independently and uniformly at random from among the $n$ coordinates of $x$, with replacement.  Define
		\[
		\mbf{T}(x)=\mbf{N}''(\hat{\mbf{x}})*E(x)+(\one-\mbf{N}'')(\hat{\mbf{x}})*\mbf{T}'(\hat{\mbf{x}}),
		\]
		where $E=(E_1,\ldots,E_m)$, with $E_i=\Ex_{[t_i-300\sqrt{t\log(1/\varepsilon)},t_i+300\sqrt{t\log(1/\varepsilon)}]}\Thr_n^{t_i},\,i\in[m]$.  We will now prove that $\mbf{T}$ is an $\varepsilon$-error probabilistic poly-tuple for $T$.
		
		\vsf
		\tsf{Correctness of Inductive Construction.}\quad  
		We now check that the construction above gives an $\varepsilon$-error probabilistic poly-tuple for $T$.  Fix any $a\in\{0,1\}^n$.  Let $\hat{\mbf{a}}\in\{0,1\}^{n/10}$ be chosen as given in the inductive construction.
		
		Suppose $|a|\le2t$.  Let $\theta=10\sqrt{t\log(1/\varepsilon)}$.  Applying Lemma \ref{lem:chernoff-absolute}, we get 
		\begin{align*}
			\prob{}{||\hat{\mbf{a}}|-|a|/10|>\theta}&\leq2\exp\left(-\frac{\theta^2}{2\cdot |a|\cdot (|a|/n)(1-|a|/n)) + (2\theta/3)}\right)\\
			&\leq 2\exp\left(-\frac{100t\log(1/\varepsilon)}{2\cdot (2t)\cdot (1/4) + 7\sqrt{t\log(1/\varepsilon)}}\right) \\
			&\leq 2\exp\left(-\frac{100t\log(1/\varepsilon)}{t + 7\sqrt{t\log(1/\varepsilon)}}\right)\\
			&\leq 2\exp\left(-\frac{100t\log(1/\varepsilon)}{2t}\right)\leq \varepsilon/4.
		\end{align*}
		where for the third inequality, we have used the fact that $\log(1/\varepsilon)\leq t/160000.$

		By induction hypothesis, the probability that $\mbf{T}'(\hat{\mbf{a}})$ does not agree with $T'(\hat{\mbf{a}})$ is at most $\varepsilon/4$, and similarly for $\mbf{T}_+''$ and $\mbf{T}_-''$. Let $\mc{G}_a$ be the event that none of the above events occur and that $||\hat{\mbf{a}}|-(|a|/10)|\leq \theta$; by a union bound, the event $\mc{G}_a$ occurs with probability at least $1-\varepsilon$. In this case, we show that $\mbf{T}(a) = T(a),$ which will prove the correctness of the construction in the case that $|a|\leq 2t$.
		
		To see this, observe the following for each $i\in [m]$.
		\begin{itemize}
			\item  $\mbf{T}_i'(\hat{\mbf{a}})=T_i(a)$ if $||a|-t_i|>10\theta$. This is because $\mbf{T}_i'(\hat{\mbf{a}}) = T_i'(\hat{\mbf{a}})$ by our assumption that the event $\mc{G}_a$ has occurred. Further, we also have $T_i'(\hat{\mbf{a}}) = T_i(a)$ since $|\hat{a}-|a|/10|\leq \theta$ (by occurrence of $\mc{G}_a$) and hence $|a|\geq t_i$ if and only if $|\hat{a}|\geq t_i/10.$
			\item  If $||a|-t_i|>30\theta$, then $\mbf{N}''_i(\hat{\mbf{a}})=0$. This is because $||\hat{a}|-|a|/10|\leq \theta$ and hence $||\hat{a}|-t_i/10|> 2\theta.$ Hence, either $\mbf{T}_{i,+}''(\hat{a}) = 1$ or $\mbf{T}_{i,-}''(\hat{a}) = 0$ and therefore, $\mbf{N}''_i(\hat{\mbf{a}}) = 0.$
			
			Thus, when $||a|-t_i| > 30\theta,$ the definition of $\mbf{T}$ yields $\mbf{T}_i(a) = \mbf{T}_i'(\hat{\mbf{a}})$ which equals $T_i(a)$ whenever $||a|-t_i|>10\theta$ as argued above. We are therefore done in this case.
			
			\item  If $||a|-t_i|<10\theta$, then $\mbf{N}_i''(\hat{\mbf{a}})=1$. This is similar to the analogous statement above. 
			
			Therefore, when $||a|-t_i|<10\theta$, we have $\mbf{T}_i(a) = E_i(a) = T_i(a)$ as $|a|\in [t_i - 300\sqrt{t\log(1/\varepsilon)},t_i + 300\sqrt{t\log(1/\varepsilon)}]$. Hence, we are done in this case also.
			
			\item  If $10\theta\le||a|-t_i|\le30\theta$, then $E_i(a)=\mbf{T}'(\hat{\mbf{a}}) = T_i(a)$. Since $\mbf{N}_i''(\hat{\mbf{a}}) \in \{0,1\}$ for each $i\in [m]$, we again obtain $\mbf{T}_i(a) = T_i(a).$ 
		\end{itemize}
		This shows that for any $a$ such that $|a|\leq 2t,$ whenever $\mc{G}_a$ does not occur, $\mbf{T}(a) = T(a)$. 
		
		Now suppose $|a|>2t$.  Then by Bernstein's inequality (Lemma \ref{lem:chernoff-absolute}), we get 
		\begin{align*}
		\prob{}{|\hat{\mbf{a}}|<1.5t/10} &\le \prob{}{||\hat{\mbf{a}}| - (|a|/10)|\geq (|a|/40)}\\
		&\le 2\exp\left(-\frac{(|a|/40)^2}{2\cdot |a|\cdot (|a|/n)(1-(|a|/n)) + (2/3)\cdot (|a|/40)}\right)\\
		&\le  2\exp\left(-\frac{(|a|/40)^2}{2\cdot |a|\cdot (1/4) + (2/3)\cdot (|a|/40)}\right)\\
		&\le 2\exp\left(-\frac{(|a|/40)^2}{(3|a|/5)}\right)\le 2\exp\left(-\frac{|a|}{960}\right)\\
		&\le 2\exp(-t/{480}) < \varepsilon/2.		
		\end{align*}

		Also, by the induction hypothesis, the probability that $\mbf{T}'(\hat{\mbf{a}})$ does not agree with $T'(\hat{\mbf{a}})$ is at most $\varepsilon/4$, and similarly for $\mbf{T}_+''$ and $\mbf{T}_-''$. Let $\mc{G}_a$ denote the event that none of the above events occur and that $|a|\geq 3t/20$; we have $\prob{}{\mc{G}_a}\geq 1-\varepsilon$. As above, we show that when $\mc{G}_a$ occurs, then $\mbf{T}(a) = T(a)$.
		
		To see this, we proceed as follows.
		\begin{itemize}
			\item Since $|a|\geq 2t$ and $|\hat{\mbf{a}}|\geq 3t/20,$ both $T(a)$ and $\mbf{T}_i'(\hat{a})$ are both the constant-1 vector.
			\item Further, we note that we have $\mbf{N}_i''(\hat{\mbf{a}}) = 0$ for each $i\in [m]$. This is because $||\hat{\mbf{a}}| - t_i/10|\geq (|\hat{\mbf{a}}| - t/10) \geq t/20 > 20\sqrt{t\log(1/\varepsilon)},$ where the final inequality uses $t > 160000\log(1/\varepsilon).$
			
			This implies that $\mbf{T}_i(a) = \mbf{T}_i'(\hat{\mbf{a}}) = 1$ for each $i\in [m]$.
		\end{itemize}
		Hence, when $\mc{G}_a$ does not occur, we have $\mbf{T}(a) = T(a)$, which proves the correctness of the construction.
		
		\vsf
		\tsf{Correctness of Degree.}\quad We need to argue that $\deg(\mbf{T})$ satisfies the inductive claim. Suppose $\charac(\F)=p>0$.  We have
		\begin{align*}
		\deg\mbf{T}&\le\deg\mbf{N}''+\max\{\deg E,\deg\mbf{T}'\}\\
		&\le\deg\mbf{T}_+''+\deg\mbf{T}_-''+\max\{\deg E,\deg\mbf{T}'\}.
		\end{align*}
		Recall that $A_p=B_p=6400000p$.    Now
		\begin{align*}
		\deg\mbf{T}_+''+\deg\mbf{T}_-''&\le A_p\Bigg(\sqrt{\frac{t}{10}+20\sqrt{t\log\bigg(\frac{1}{\varepsilon}\bigg)}}+\sqrt{\frac{t}{10}-20\sqrt{t\log\bigg(\frac{1}{\varepsilon}\bigg)}}\Bigg)\sqrt{\log\bigg(\frac{4}{\varepsilon}\bigg)}\\
		&\quad+2B_p\log\bigg(\frac{4}{\varepsilon}\bigg)  \\
		&\le A_p\sqrt{\bigg(\frac{t}{5}+2\sqrt{\frac{t^2}{100}-400t\log\bigg(\frac{1}{\varepsilon}\bigg)}\bigg)\log\bigg(\frac{4}{\varepsilon}\bigg)}+2B_p\log\bigg(\frac{4}{\varepsilon}\bigg)\\
		&\le A_p\sqrt{\bigg(\frac{t}{5}+2\sqrt{\frac{t^2}{100}-\frac{t^2}{400}}\bigg)\log\bigg(\frac{4}{\varepsilon}\bigg)}+2B_p\log\bigg(\frac{4}{\varepsilon}\bigg)\\
		&\leq A_p\sqrt{\bigg(\frac{2+\sqrt{3}}{10}\bigg)t\log\bigg(\frac{4}{\varepsilon}\bigg)}+2B_p\log\bigg(\frac{4}{\varepsilon}\bigg)\\
		&\leq A_p\sqrt{\frac{38t}{100}\log\left(\frac{4}{\varepsilon}\right)} + 2B_p \log\left(\frac{4}{\varepsilon}\right),
		\end{align*}
		and
		\begin{align*}
		\max\{\deg E,\deg\mbf{T}'\}&\le\max\bigg\{600\sqrt{t\log\bigg(\frac{1}{\varepsilon}\bigg)},A_p\sqrt{\frac{t}{10}\log\bigg(\frac{4}{\varepsilon}\bigg)}+B_p\log\bigg(\frac{4}{\varepsilon}\bigg)\bigg\}\\
		&=A_p\sqrt{\frac{t}{10}\log\bigg(\frac{4}{\varepsilon}\bigg)}+B_p\log\bigg(\frac{4}{\varepsilon}\bigg).
		\end{align*}
		So we get
		\begin{align*}
		\deg\mbf{T}&\le \bigg(\sqrt{\frac{38}{100}}+\sqrt{\frac{1}{10}}\bigg)A_p\sqrt{t\log\bigg(\frac{4}{\varepsilon}\bigg)}+3B_p\log\bigg(\frac{4}{\varepsilon}\bigg)\\
		&\le\frac{94}{100}A_p\sqrt{\left(t\log\bigg(\frac{1}{\varepsilon}\bigg)+2t\right)}+3B_p\log\bigg(\frac{1}{\varepsilon}\bigg) + 6B_p\\
		&\le\frac{95}{100}A_p\sqrt{t\log\bigg(\frac{1}{\varepsilon}\bigg)}+4B_p\log\bigg(\frac{1}{\varepsilon}\bigg)\\
		&\le\frac{95}{100}A_p\sqrt{t\log\bigg(\frac{1}{\varepsilon}\bigg)}+3B_p\log\bigg(\frac{1}{\varepsilon}\bigg)+B_p\log\bigg(\frac{1}{\varepsilon}\bigg)\\
		&\le A_p\sqrt{t\log\bigg(\frac{1}{\varepsilon}\bigg)}+B_p\log\bigg(\frac{1}{\varepsilon}\bigg).
		\end{align*}
		where the third inequality uses $\varepsilon \leq 2^{-100}$ and the final inequality uses $t > 160000\log(1/\varepsilon).$
		
		Now if $\charac(\F)=0$, then we get a similar degree bound with $A_0=B_0=64000000$.  This completes the argument for correctness of degree.	
	\end{proof}

	\subsection{Upper bound on $\pdeg_\varepsilon(g)$}
	\label{sec:Lu}
	
	This result is due to Lu~\cite{Lu} but a proof is sketched here for completeness.
	
	Recall that $\charac(\F) = p.$
	
	When $\per(g) = 1,$ $g$ is a constant function and hence the result is trivial. So assume that $\per(g) = p^t$ for $t\ge 1$. In this case, we show that $g$ can be represented \emph{exactly} as a linear combination of elementary symmetric polynomials of degree at most $D = p^t-1,$ which clearly proves the upper bound stated in the theorem. To see that every such $g$ has such a representation, we proceed as follows. 
	
	Let $V$ be the vector space generated by all functions $f:\{0,1\}^n\rightarrow \F$ such that $f$ can be written as a linear combination of elementary symmetric polynomials of degree at most $D$. Clearly, since there is a $1$-$1$ correspondence between multilinear polynomials and functions from $\{0,1\}^n$ to $\F,$ the vector space $V$ has dimension exactly $D+1=p^t$. Each function in $f$ is a symmetric (not necessarily Boolean) function on $\{0,1\}^n$. Further, a standard application of Lucas' theorem (see~\cite{Lucas-thm}) shows that each $f\in V$ satisfies 
	\begin{equation}
	\label{eq:propV}
	\spec f(i) = \spec f(i+p^t)
	\end{equation}
	for each $i\leq n-p^t$.
	
	Now, consider the vector space $W$ of all functions $f':\{0,1\}^n\rightarrow \F$ that satisfy property (\ref{eq:propV}). Clearly, $W\subseteq V.$ Furthermore, the functions $\MOD_{p^t}^i$ ($i\in [0,p^t-1]$) is a set of $p^t$ many linearly independent functions in $W$. Hence, the dimension of $W$ is exactly $p^t$ and therefore, $W=V.$
	
	Since $g\in W$, we immediately see that $g\in V$ and is hence a linear combination of elementary symmetric polynomials of degree at most $D$.

	\subsection{Upper bound for $\pdeg_\varepsilon(h)$.} 
	
	Let $B(h)=k$.  Thus we can write $h=h_1+(1-\widetilde{h_2})$, for $k$-constant symmetric Boolean functions $h_1,h_2$, where $\widetilde{h_2}(x_1,\ldots,x_n)=h_2(1-x_1,\ldots,1-x_n)$.  But then by Corollary~\ref{cor:pdeg-t-bdd}, $\pdeg_\varepsilon(h_1)=\pdeg_\varepsilon(h_2)=O(\sqrt{k\log(1/\varepsilon)}+\log(1/\varepsilon))$ and so $\pdeg_\varepsilon(h)=O(\sqrt{k\log(1/\varepsilon)}+\log(1/\varepsilon))$ over positive characteristic $p$. For $p=0$, we obtain the same upper bound up to log-factors.
	
	\subsection{Upper bound for $\pdeg_\varepsilon(f)$.} Let $(g,h)$ be the standard decomposition of $f$.  So $f=g\oplus h=g+h-2gh$.  Further, we already have the Alman-Williams bound of $O(\sqrt{n\log(1/\varepsilon)})$ on $\pdeg_\varepsilon(f)$ (Lemma~\ref{lem:AW}).  So we get $\pdeg_\varepsilon(f)=O(\min\{\sqrt{n\log(1/\varepsilon)},\per(g)+\sqrt{B(h)\log(1/\varepsilon)}+\log(1/\varepsilon)\})$ over positive characteristic and the same bound up to log-factors over characteristic $0$. This concludes the proof of Theorem~\ref{thm:main-ubd}.

	\section{Lower Bounds}
	\label{sec:lbd}
	
	We now prove the lower bounds given in Theorem \ref{thm:main-lbd}.  Throughout this section, let $\F$ be any field. We use $\pdeg_\varepsilon(\cdot)$ instead of $\pdeg^\F_\varepsilon(\cdot).$
	
	\subparagraph*{High-level outline of proof.} We use a similar proof strategy to Lu~\cite{Lu}, who gave a characterization of symmetric functions computable by quasipolynomial-sized $\AC^0[p]$ circuits. To prove a lower bound on the probabilistic degree of a symmetric function $f\in \sB_n,$ we will use known lower bounds for Majority (Lemma~\ref{lem:majlbd}) and $\MOD_q$ functions (Lemma~\ref{lem:modlbd}), where $q$ is relatively prime to the characteristic $p$. The basic idea is to use a few ``restrictions'' of $f$ to compute either a Majority function or a $\MOD_q$ function. 
	
	Here, a restriction of $f$ is a function $h\in \sB_m$ obtained by setting a few inputs of $f$ to $0$s and $1$s. That is, we define $h(x) = f(x0^a1^{n-m-a})$ for some $a$. From the definition of such an $h$, it is clear that for any $\delta > 0,$ $\pdeg_\delta(h)\leq \pdeg_\delta(f).$ We design a small number of such restrictions $h_1,\ldots,h_\ell$ and a ``combining function'' $P:\{0,1\}^\ell\rightarrow \{0,1\}$ such that either Majority of $\MOD_q$ can be written as $P(h_1,\ldots,h_\ell).$
	
	The main restriction we will have on $P$ is that it should be a low-degree polynomial.\footnote{This is the point of divergence from Lu's result. His constraint on the function $P$ was that it have a small $\AC^0[p]$ circuit. Our focus is different.} Given this, using Fact~\ref{fac:pdeg}, we can write
	\[
	\pdeg_\varepsilon(P(h_1,\ldots,h_\ell)) \leq \deg(P)\cdot \max_{i}\pdeg_{\varepsilon/\ell}(h_i) \leq \deg(P)\cdot \pdeg_{\varepsilon/\ell}(f) \leq \pdeg_{\varepsilon}(h) \cdot (\log b\cdot \log \ell).
	\]
	Using lower bounds on the probabilistic degree of Majority and $\MOD_q$, we then get lower bounds on the probabilistic degree of $f$.
	
	The non-trivial part is to determine the hard function we use in the reduction and how to carry out the reduction with a polynomial $P$ of low degree. Both of these are dependent on the structure of $\spec f.$ We give the details below.
	
	We start with a preliminary lemma.
	
	\begin{lemma}\label{lem:improv}
		Let $\mc{F}$ be a set of functions mapping $[0,n-1]$ to $\{0,1\}$ such that
		\begin{itemize}
			\item  $1-f\in\mc{F}$, for all $f\in\mc{F}$.
			\item  For every $i,j\in[0,n-1],\,i\ne j$, there exists $f\in\mc{F}$ such that $f(i)=1,\,f(j)=0$.
		\end{itemize}
		Then there exists $S\subseteq\mc{F}$ such that $|S|\le\log n$ and $f_S=\prod_{f\in S}f$ has support size 1.\footnote{Here, the product between functions mapping $[0,n-1]$ to $\F$ is defined pointwise.}
	\end{lemma}
	\begin{proof}
		It is enough to prove that for every positive integer $k\leq \lfloor \log n\rfloor $, there exists $S_k\subseteq\mc{F}$ such that $|S|\le k$ and $1\le|\supp(f_{S_k})|\le n/2^k$.  We do so by induction on $k$.
		
		If $k=1$, consider any $f\in\mc{F}$.  If $|\supp(f)|\le n/2$, then we choose $S=\{f\}$; otherwise we have $|\supp(1-f)|\le n/2$ and then we choose $S=\{1-f\}$.  Now consider any $k\in[\lfloor \log n\rfloor-1]$ and assume the existence of $S_k$.  If $|\supp(f_{S_k})|=1$, then we choose $S_{k+1}=S_k$, which satisfies the required conditions.  Now suppose $|\supp(f_{S_k})|>1$.  Choose any $i,j\in\supp(f_{S_k}),\,i\ne j$.  By the given condition, there exists $f_{i,j}\in\mc{F}$ such that $f_{i,j}(i)=1$ and $f_{i,j}(j)=0$.  If $|\supp(f_{i,j}f_{S_k})|\le|\supp(f_{S_k})|/2$, then $|\supp(f_{i,j}f_{S_k})|\le n/2^{k+1}$ and so we choose $S_{k+1}=S_k\cup\{f_{i,j}\}$; otherwise we have $|\supp((1-f_{i,j})f_{S_k})|\le|\supp(f_{S_k})|\le n/2^{k+1}$ and so we choose $S_{k+1}=S_k\cup\{1-f_{i,j}\}$.
		
		At the end of this process, we have a set $S$ such that $1\leq |\supp(f_S)|\leq n/2^{\lfloor \log n\rfloor} < 2$ and hence $|\supp(f_S)|=1.$
	\end{proof}

%
%
	
	\begin{remark}\label{rem:improv}
 The above lemma is also true if we take any interval \(I=[a,a+m-1]\) in \(\mb{Z}\) instead of \([0,m-1]\).  This can be checked by further taking the bijection \(\phi_a:[a,a+m-1]\to[0,m-1]\) defined as \(\phi(x)=x-a,\,x\in[a,a+m-1]\).
	\end{remark}
		
	\subsection{Lower bound on \(\pdeg_\varepsilon(g)\)}
	
	We start with a simple consequence of Lemma~\ref{lem:improv}.
	
	Let $u:[0,m-1]\rightarrow \{0,1\}$ be any function. Let $\tau_m: [0,m-1]\rightarrow [0,m-1]$ be the cyclic-shift operator: that is, $\tau_m(i) = i+1$ for $i < m-1$, and $\tau_m(m-1) = 0.$ We say that $u$ is \emph{periodic} if $u\circ \tau^j_m = u$ for some $j \in [m-1]$, and \emph{aperiodic} otherwise. 
	
	\begin{lemma}\label{lem:circ-improv}
	Let $\F$ be any field. Fix any \(u:[0,m-1]\to\{0,1\}\) that is aperiodic.  Let $u_j=f\circ\tau^j_m$ for $j\in[0,m-1]$.  For any $g:[0,m-1]\rightarrow \F$, there is a $P\in \F[Y_0,\ldots,Y_m-1]$ with degree at most $\log m$ such that $g = P(u_0,\ldots,u_{m-1})$.
	\end{lemma}
	\begin{proof}
	For any $i\in [0,m-1]$, let $\delta_i:[0,m-1]\rightarrow \{0,1\}$ be the delta function supported only at $i$. By linearity, it suffices to show the lemma for $g\in \{\delta_0,\ldots,\delta_{m-1}\}$. Further, by symmetry, it suffices to show that there is at least one $i$ such that $\delta_i$ is a polynomial of degree at most $\log m$ in $u_0,\ldots, u_{m-1}.$
	
	To prove the latter, we use Lemma~\ref{lem:improv}.  Let \(\mc{F}=\{u_j:j\in[0,m-1]\}\cup\{1-u_j:j\in[0,m-1]\}\).  We first note that $\mc{F}$ satisfies the hypotheses of Lemma~\ref{lem:improv}. Clearly, $v\in \mc{F}$ implies $1-v\in \mc{F}.$  Now consider any \(i,j\in[0,m-1],\,i \ne j\).  If \(v(i)=v(j)\) for all \(v\in\mc{F}\), then in particular, we have \(u_k(i)=u_k(j)\), for all \(k\in[0,m-1]\).  This implies \(u\) is periodic, which is a contradiction.  Thus there exists \(v\in\mc{F}\) such that \(v(i)=0,\,v(j)=1\).  So by Lemma \ref{lem:improv}, there exists \(S\subseteq\mc{F},\,|S|\le\log m\) such that \(\prod_{v\in S}v\) has support size 1.  We have \(S=\{\hat u_{i_1},\ldots,\hat u_{i_t}\}\), where for every \(k\in[t]\), \(\hat u_{i_k}\) is either \(u_{i_k}\) or \(1-u_{i_k}\).  So now for \(k\in[t]\), define
		\[
		M_k(Y_0,\ldots,Y_{m-1})=\begin{cases}Y_{i_k},&\hat u_{i_k}=u_{i_k}\\
		1-Y_{i_k},&\hat u_{i_k}=1-u_{i_k}\end{cases}
		\]
		Then we have the required polynomial as \(P(Y_0,\ldots,Y_{m-1})=\prod_{k\in[t]}M_k(Y_0,\ldots,Y_{m-1})\).  Since \(\prod_{v\in S}v\) has support size 1, we have \(P(u_0,\ldots,u_{m-1})=\delta_i\), for some \(i\in[0,m-1]\).
	\end{proof}
	
	We now prove the lower bound for $\pdeg_\varepsilon(g).$ 
	\begin{lemma}\label{lem:pdeg-g-improv}
	For any $\varepsilon\in [2^{-n},1/3],$
		\[
		\pdeg_\varepsilon(g)\ge\begin{cases}
		\Omega\left(\frac{\sqrt{n\log(1/\varepsilon)}}{\log^2 n}\right),&\per(g)\tx{ is not a power of }p\\
		\Omega\left(\frac{\min\{\per(g),\sqrt{n\log(1/\varepsilon)}\}}{\log^2 n}\right),&\per(g)\tx{ is a power of }p.
		\end{cases}
		\]
	\end{lemma}

	\begin{proof}
	By Fact~\ref{fac:pdeg} item 1, we know that $\pdeg_\varepsilon(g) = \Theta(\pdeg_\delta(g))$ as long as $\delta = \varepsilon^{\Theta(1)}.$ In particular, we may assume without loss of generality that $\varepsilon\in [2^{-n/10},1/5].$
	
	 Let $b$ denote $\per(g)$. We know (Observation~\ref{obs:decomp}) that \(b := \per(g)\le \lfloor n/3\rfloor\).  
	
	We have two cases.
		
		\subparagraph*{\(b\) is not a power of \(p\).}  Define the function $u:[0,b-1]\rightarrow \{0,1\}$ by $u(i) = \spec g (i).$ Note that for $u_j = u\circ \tau^j_b$ ($j\in [0,b-1]$), as defined in the statement of Lemma~\ref{lem:circ-improv}, we have $u_j = \spec g(j+i)$ (the latter is well defined as $j+i < 2b <  n$). Further, as $b$ is the period of $g$, Corollary~\ref{cor:string} implies that $u \neq u_j$ for all $j\in [b-1]$. This means that $u$ is aperiodic (as defined above).
		
		 Let $q$ be any prime divisor of $b$ distinct from $p$. For each $i\in [0,q-1]$, define a function $v_i:[0,b-1]\rightarrow \{0,1\}$ by $v_i(j) = 1$ iff $j\equiv i \pmod{q}.$ Lemma~\ref{lem:circ-improv} implies that for each $i$, there is a $P_i(Y_0,\ldots,Y_{b-1})$ of degree at most $\log b$ such that $P_i(u_0,\ldots,u_{b-1}) = v_i.$
		
		Fix any $i\in [0,b-1]$ and consider the function $G_i:\{0,1\}^{n-b}\rightarrow \F$ defined by $G_i(x) = P_i(g(x0^b), g(x0^{b-1}1),\ldots, g(x0 1^{b-1})).$ Clearly, as all the inputs to $P_i$ are $b$-periodic symmetric functions, the same holds for the function $G_i$. Further, for any  $j\in [0,b-1]$, we see that 
		\begin{align*}
		\spec G_i(j) &= P_i(\spec g(j), \ldots, \spec g(j+b-1)) \\
		&= P_i(\spec u_0(j),\ldots, \spec u_{b-1}(j))\\
		&=\left\{\begin{array}{ll}
		1 & \text{if $j\equiv i \pmod{q}$,}\\
		0 & \text{otherwise.}
		\end{array}\right.
		\end{align*}
		This implies that $G_i$ is in fact the $\MOD^{q,i}_{n-b}$ function. Note also that $q\leq b \leq (n-b)/2$. This will be relevant below as we will apply Lemma~\ref{lem:modlbd} to one of the functions $G_0,\ldots,G_{q-1}$.
		
		By Fact~\ref{fac:pdeg}, for any $\delta > 0$, we have that 
		\begin{align*}
		\pdeg_{\delta}(G_i) &\leq (\log b)\cdot \max_{i\in [0,b-1]} \pdeg_{\delta/b}(g(x0^{b-i}1^i))\\ 
		&\leq (\log b)\cdot \pdeg_{\delta/b}(g) \leq (\log b)\cdot \pdeg_{\varepsilon}(g)\cdot O\left(\frac{\log(b/\delta)}{\log(1/\varepsilon)}\right).
		\end{align*}
		
		In particular, setting $\delta = \varepsilon/q,$ we have $\pdeg_{\varepsilon/q}(G_i) \leq \pdeg_{\varepsilon}(g) \cdot O(\log^2 b)$ for each $i\in [0,b-1]$. On the other hand, Lemma~\ref{lem:modlbd} implies that for some $i\in [0,q-1]$, $\pdeg_{\varepsilon/b}(G_i) = \Omega(\sqrt{n\log(1/\varepsilon)}).$ Putting these together, we obtain the claimed lower bound on $\pdeg_{\varepsilon}(g).$

		\subparagraph*{\(b\) is a power of \(p\).} In this case, we first choose parameters $m,\delta$ with the following properties.
		\begin{enumerate}
		\item[(P1)] $m\in [n-b]$ and $m\equiv (n-b)\pmod{2}.$
		\item[(P2)] $(1/5)\geq \delta \geq \max\{\varepsilon,1/2^{m}\}.$
		\item[(P3)] $4\sqrt{m\log(1/\delta)}\leq b.$
		\item[(P4)] $\sqrt{m\log(1/\delta)} = \Omega(\min\{b,\sqrt{n\log(1/\varepsilon)}\}).$
		\end{enumerate}
		
		We will show later how to find $m,\delta$ satisfying these properties. Assuming this for now, we first prove the lower bound on $\pdeg_\varepsilon(g).$ First, we fix some $b$-periodic symmetric function $G:\{0,1\}^{n-b}\rightarrow \{0,1\}$ so that $G$ agrees with $\Maj_{n-b}$ on all inputs of weight $a \in ((n-b)/2 - 2\sqrt{m\log(1/\delta)}, (n-b)/2 + 2\sqrt{m\log(1/\delta)})$ (it is possible to define such a $b$-periodic $G$ because of property (P3) above). 
		
		As in the previous case, we can use Lemma~\ref{lem:circ-improv} to find a polynomial $P\in \F[Y_0,\ldots,Y_{b-1}]$ of degree at most $\log b$ so that $\spec G(j) = P(\spec g (j), \spec g(j+1),\ldots,\spec g(j+(b-1)))$ for all $j\in [0,n-b].$ (We omit the proof as it is very similar.) This implies that for any $x\in \{0,1\}^{n-b}$,
		\[
		G(x) = P(g(x0^b), g(x0^{b-1}1),\ldots,g(x0 1^{b-1})).
		\]
		In particular, $\pdeg_{\delta}(G) \leq (\log b)\cdot \pdeg_{\delta/b}(g)\leq O(\log^2 b)\cdot \pdeg_{\delta}(g).$
		
		Now, consider $G':\{0,1\}^m\rightarrow \{0,1\}$ defined by $G'(x) = G(x0^t 1^t)$ where $t = (n-b-m)/2$ (note that (P1) implies that $(n-b-m)$ is even). Clearly, $G'$ agrees with $G$, and hence $\Maj_m$, on inputs of Hamming weight $a\in (m/2 - 2\sqrt{m\log(1/\delta)}, m/2 + 2\sqrt{m\log(1/\delta)})$. By the Bernstein inequality (Lemma~\ref{lem:chernoff-absolute}), it follows that $G'$ agrees with $\Maj_m$ on at least a $(1-\delta)$ fraction of its inputs. Hence, by Lemma~\ref{lem:majlbd}, it follows that $\pdeg_{\delta}(G') = \Omega(\sqrt{m\log(1/\delta)}).$ 
		
		However, we have $\pdeg_{\delta}(G')\leq \pdeg_{\delta}(G)$ which in turn is bounded by $O(\log^2 b)\cdot \pdeg_{\delta}(g)$ as argued above. So we obtain
		\[
		\pdeg_{\varepsilon}(g) \geq \pdeg_{\delta}(g) = \Omega\left(\frac{\sqrt{m\log(1/\delta)}}{\log^2 b}\right) = \Omega\left(\frac{\min\{b,\sqrt{n\log(1/\varepsilon)}\}}{\log^2 n}\right) 
		\]
		where the first inequality follows from the fact that $\delta \geq \varepsilon$ (by (P2)), and the second equality follows from property (P4) above and the fact that $b\leq n.$ 
		
		It remains to show that we can choose $m,\delta$ satisfying (P1)-(P4) as above. This we do as follows. 
		\begin{enumerate}
		\item If $b\leq 10\sqrt{n},$ we take $m$ to be the largest integer satisfying (P1) and such that $m\leq b^2/100.$ The parameter $\delta$ is set to $1/5.$ 
		\item If $b\geq n/10,$ we take $m $ to be the largest integer satisfying (P1) and such that $m\leq n/100$, and set $\delta =\max\{\varepsilon, 2^{-m}\}$.
		\item Finally, if $10\sqrt{n} < b < n/10,$ then we take $m = n-b$  and $\delta = \max\{\varepsilon,2^{-b^2/16m}\}.$
		\end{enumerate}
		In each case, the verification of properties (P1)-(P4) is a routine computation. (We assume throughout that $b$ is greater than a suitably large constant, since otherwise the statement of the lemma is trivial.) This concludes the proof.
	\end{proof}
	
	\subsection{Lower bound on \(\pdeg_\varepsilon(h)\)}
	
	We start with the special case of thresholds.
	
	\begin{lemma}\label{lem:thr-lbd}
	Assume $1\leq t\leq  n/2 $. For any $\varepsilon\in [2^{-n},1/3],$
	\[
	\pdeg_{\varepsilon}(\Thr_n^t)= \Omega(\sqrt{t\log(1/\varepsilon)} + \log(1/\varepsilon)).
	\]
	\end{lemma}
	\begin{proof}
     By Fact~\ref{fac:pdeg} item 1, we can assume that $\varepsilon\in [2^{-n/2},1/5]$. The proof breaks into two cases depending on the relative magnitudes of $\varepsilon$ and $2^{-t}$.
 	
	Consider the case when $\varepsilon\geq 1/2^t.$ In this case, it suffices to show that $\pdeg_{\varepsilon}(\Thr_n^t)= \Omega(\sqrt{t\log(1/\varepsilon)}).$ Note that $\Maj_{2t-1}(x) = \Thr_{2t-1}^t(x) = \Thr_n^t(x0^{n-2t+1})$ for any $x\in \{0,1\}^{2t-1}$. Therefore, we have by Lemma~\ref{lem:majlbd} that
	\[
	\pdeg_{\varepsilon}(\Thr_n^t)\geq \pdeg_{\varepsilon}(\Maj_{2t-1}) = \Omega(\sqrt{t\log(1/\varepsilon)}).
	\]
	
	Now, consider the case when $\varepsilon < 2^{-t}$. Now, it suffices to show that $\pdeg_{\varepsilon}(\Thr_n^t)= \Omega(\log(1/\varepsilon)).$ Note that we have $\OR_{\lceil n/2\rceil}(x) = \Thr_{\lceil n/2\rceil}^1(x) = \Thr_n^t(x1^{t-1}0^{\lfloor n/2\rfloor - t+1})$, and so the lemma is implied by the following statement. For any positive integer $m$ and $\varepsilon\in [2^{-m},1/3]$
	\begin{equation}
	\label{eq:OR-lbd}
	\pdeg_{\varepsilon}(\OR_m)= \Omega(\log(1/\varepsilon)).
	\end{equation}
	
	While the above is possibly folklore, we don't know of a reference with a proof, so we give one here. Assume $m,\varepsilon$ as above and let $D$ denote $\pdeg_{\varepsilon}(\OR_m).$ By setting some bits to $0$, we also get $\pdeg_{\varepsilon}(\OR_{m_1})\leq D$ where $m_1=\lfloor \log(1/\varepsilon) - 1\rfloor.$ Fix such a probabilistic polynomial $\bf{P}$ of degree at most $D$ for $\OR_{m_1}$. We have for every $x\in \{0,1\}^{m_1},$
	\[
	\prob{\mbf{P}}{\mbf{P}(x) \neq \OR_{m_1}(x)} \leq \varepsilon < \frac{1}{2^{m_1}}.
	\] 
	By a union bound, there is some polynomial $P$ in the support of the probability distribution underlying $\mbf{P}$ that agrees with the function $\OR_{m_1}$ everywhere. However, the unique multilinear polynomial representing $\OR_{m_1}$ has degree $m_1$. Hence, we see that $\deg(\mbf{P}) \geq \deg(P) = m_1 = \Omega(\log(1/\varepsilon))$ concluding the proof of (\ref{eq:OR-lbd}).
	\end{proof}
	
	We now prove the lower bound on $\pdeg_\varepsilon(h)$ from Theorem~\ref{thm:main-lbd}.
	
	\begin{lemma}\label{lem:pdeg-h-improv}
		Assume $B(h)\geq 1$. Then, $\varepsilon\in [2^{-n},1/3],$
		\[
		\pdeg_\varepsilon(h)=\Omega\left(\frac{\sqrt{B(h)\log(1/\varepsilon)}+\log(1/\varepsilon)}{\log n}\right).
		\]
	\end{lemma}
	\begin{proof}
		As in Lemma~\ref{lem:pdeg-g-improv}, we can assume that $\varepsilon\in [2^{-n/10},1/5].$	
	
		By Observation~\ref{obs:decomp}, we have \(b:=B(h)\le \lceil n/3\rceil\) and further, that either $\spec h(b-1)=1$ or \(\spec h(n-b+1)=1\).  We assume that $\spec h(b-1)=1$ (the other case is similar). 
		
		Let $m = \lfloor n/6 \rfloor$ and $t = \lceil b/3\rceil .$ Note that $t\leq m$. 
		
		For \(i\in[0,t-1]\), define \(h_i\in\sB_{m+t}\) by
		\[
		h_i(x)=h(x1^{b-t+i}0^{n-m-b-i}).
		\]
		Then for every \(i\in[0,t-1]\), we have \(\spec h_i=y_i10^{m+i+1}\), for some \(y_i\in\{0,1\}^{t-1-i}\).  By standard linear algebra, it follows that the vector $1^t0^{m+1}$ is in the span of the vectors $\spec h_i$ ($i\in [0,t-1]$). It follows that we can write $1-\Thr_{m+t}^t = \sum_{i=0}^{t-1}\alpha_i h_i$ for some choice of $\alpha_0,\ldots,\alpha_{t-1}\in\F.$  So we have
		\[
		\pdeg_\varepsilon(\Thr_{m+t}^t)\le\max_{i}\pdeg_{\varepsilon/t}(h_i)\le \pdeg_{\varepsilon/t}(h)\leq O(\log b)\cdot \pdeg_{\varepsilon}(h).
		\]
		Lemma~\ref{lem:thr-lbd} now implies the lower bound.
	\end{proof}
	
	\subsection{Lower bound on \(\pdeg_\varepsilon(f)\)}
	
	We start with a slightly weaker lower bound on $\pdeg_{\varepsilon}(f)$ that is independent of $h$.
	
	\begin{lemma}\label{lem:preg-f-weak} For any $\varepsilon\in [2^{-n},1/3],$
		\[
		\pdeg_\varepsilon(f)\ge\left\{
			\begin{array}{ll}
			\Omega\left(\frac{\sqrt{n\log(1/\varepsilon)}}{\log^2 n}\right) & \text{if $\per(g) > 1$ and not a power of $p$,}\\
			\Omega\left(\frac{\min\{\sqrt{n\log(1/\varepsilon)}, \per(g)\}}{\log^2 n}\right) & \text{if $\per(g)$ a power of $p$.}\\
			\end{array}\right.
		\]
	\end{lemma}
	
	\begin{proof}
	Similar to Lemma~\ref{lem:pdeg-g-improv}, we may assume  that $\varepsilon\in [2^{-n/100},1/5].$
	 
	The proof of this lemma splits into two cases depending on the magnitude of $b := \per(g)$. Let $n_1 = n-2\lceil n/3\rceil$.
	
	Assume first that  $b\leq \lfloor n_1/3\rfloor.$  Define $f'\in \sB_{n_1}$ by $f'(x) = f(x0^{\lceil n/3\rceil}1^{\lceil n/3\rceil}).$ By our choice of the function $g$, the function $f'$ is also a function with period $b$. Now, the proof of Lemma~\ref{lem:pdeg-g-improv} shows that $\pdeg_\varepsilon(f')$ is  $\widetilde{\Omega}(\sqrt{n\log(1/\varepsilon)})$  if $b > 1$ and not a power of $p$, and $\widetilde{\Omega}(\min\{\sqrt{n\log(1/\varepsilon)},b\})$ if $b$ is a power of $p$. The same lower bound immediately applies to $f$ also, and hence the lemma is proved in this case.
	
	From now on, we assume that $b > m:= \lfloor n_1/3\rfloor$. In particular, this implies that there is no periodic symmetric function in $\sB_n$ with period at most $m$ that agrees with $f$ on inputs of weight in $I := [\lceil n/3\rceil, \lfloor 2n/3\rfloor].$ Thus, for each $k \in [m]$, there exist $r_k\in I$ such that $r_k + k \in I$ and $\spec f(r_k) \neq \spec f(r_k + k).$
	
	Now, define a set of functions from $\sB_{3m}$ as follows. For each $i,j\in [m+1,2m]$ with $i < j$, define $f_{i,j}(x)$ as follows. Set $k = j-i$. For $r_k$ as defined above, let
	\begin{equation}
	\label{eq:def-fij}
	f_{i,j}(x) = f(x1^{r_k-i}0^{n-3m-r_k+i}).
	\end{equation}
	The parameters above are chosen so that $\spec f_{i,j}(i) = \spec f(r_k)$ and $\spec f_{i,j}(j) =\spec f(r_k + k).$ Consequently, we have $\spec f_{i,j}(i) \neq \spec f_{i,j}(j).$ Let $u_{i,j}:[m+1,2m]\rightarrow \{0,1\}$ denote the restriction of the function $\spec f_{i,j}$ to the interval $[m+1,2m].$ We denote by $\mc{U}$ the set $\{u_{i,j},1-u_{i,j}\ |\ m+1 \leq i < j\leq 2m\}$. 
	
	By Lemma~\ref{lem:improv}, we know that there is a subset $\mc{U}'\subseteq \mc{U}$ such that $s:=|\mc{U}'| = O(\log m)$ and $\prod_{u\in \mc{U}'}u$ has support $\{a\}$ for some $a\in [m+1,2m]$. Assume that $\mc{U}' = \{\hat{u}_{i_1,j_1},\ldots,\hat{u}_{i_s,j_s}\}$ where each $\hat{u}_{i_t,j_t}$ is either $u_{i_t,j_t}$ or $1-u_{i_t,j_t}$ for $t\in [s]$. 
	
	Define $\hat{f}_{i_t,j_t}$ to be $f_{i_t,j_t}$ if $\hat{u}_{i_t,j_t} = u_{i_t,j_t}$, and $1-f_{i_t,j_t}$ otherwise. Let $\mc{F}' = \{\hat{f}_{i_t,j_t}\ |\ t\in [s]\}.$ It follows from the properties of $\mc{U}'$ that $G:=\prod_{F\in \mc{F}'}F\in \sB_{3m}$ satisfies 
	\[
	\spec G(a') = \left\{
	\begin{array}{ll}
	1 & \text{if $a' = a$,}\\
	0 & \text{if $a' \in [m+1,2m]\setminus \{a\}$.}
	\end{array}\right.
	\]
	
	We will now use $G$ to construct the Majority function on $\Theta(m)$ inputs. We assume that $a \leq 3m/2$ (the other case is similar). Let $m_1 = \lfloor m/2\rfloor$ and define $G_i\in \sB_{m_1}$ for $i\in [0,m_1]$ by 
	\[
	G_i(x) = G(x1^{a-i}0^{3m-a+i-m_1}).
	\]
	Note that $\spec G_i = y_i10^{m_1-i-1}$ for some $y_i\in \{0,1\}^{i}.$ In particular, the vectors $\spec G_i\in \{0,1\}^{m_1}$ are linearly independent and hence span $\F^{m_1+1}.$ As a result, we can write $\spec \Maj_{m_1} = \sum_{i=0} \alpha_i \cdot (\spec G_i)$ for some choice of $\alpha_0,\ldots,\alpha_{m_1}\in \F.$ Equivalently, we have $\Maj_{m_1} = \sum_{i=0}^{m_1}\alpha_i G_i.$
	
	This implies that 
	\[
	\pdeg_{\varepsilon}(\Maj_{m_1}) \leq \max_{i\in [0,m_1]} \pdeg_{\varepsilon/m_1}(G_i) \leq \pdeg_{\varepsilon/m}(G).
	\]
	
	The function $G$ in turn is a product of $s = O(\log m)$ many functions from $\mc{F}'$. By (\ref{eq:def-fij}), each $f'\in \mc{F}'$ satisfies $\pdeg_{\delta}(f')\leq \pdeg_{\delta}(f)$ for any $\delta > 0$. Hence, we have 
	\begin{align*}
	\pdeg_{\varepsilon}(\Maj_{m_1}) &\leq \pdeg_{\varepsilon/m}(G)\\
	&\leq O(\log m)\cdot \max_{m+1\leq i < j \leq 2m, f'\in \mc{F}'}\pdeg_{\varepsilon/(m\log m)}(f')\\
	& \leq O(\log m)\cdot \pdeg_{\varepsilon/(m\log m)}(f) \leq O(\log^2 m) \cdot \pdeg_{\varepsilon}(f).
	\end{align*}
	
	As $\varepsilon \geq 2^{-n/100}\geq 2^{-m_1},$ Lemma~\ref{lem:majlbd} implies that $\pdeg_\varepsilon(\Maj_{m_1}) = \Omega(\sqrt{m_1\log(1/\varepsilon)}) = \Omega(\sqrt{n\log(1/\varepsilon)}).$ Along with the above inequality, this implies the desired lower bound on $\pdeg_{\varepsilon}(f).$
	\end{proof}
	
	We are now ready to prove the final lower bound on $\pdeg_\varepsilon(f).$
	
	\begin{lemma}\label{lem:preg-f-improv} For any $\varepsilon\in [2^{-n},1/3],$
		\[
		\pdeg_\varepsilon(f)\ge\left\{
			\begin{array}{ll}
			\widetilde{\Omega}(\sqrt{n\log(1/\varepsilon)}) & \text{if $\per(g) > 1$ and not a power of $p$,}\\
			\widetilde{\Omega}(\min\{\sqrt{n\log(1/\varepsilon)},\per(g)\}) & \text{if $\per(g)$ a power of $p$ and $B(h) =0$,}\\
			\widetilde{\Omega}(\min\{\sqrt{n\log(1/\varepsilon)},\per(g)  & \text{otherwise.}\\
			\ \ + \sqrt{B(h)\log(1/\varepsilon)} +\log(1/\varepsilon)\}) 
			\end{array}\right.
		\]
	\end{lemma}

	\begin{proof}
	Lemma~\ref{lem:preg-f-weak} already implies the result in the case that any of the following conditions hold.
	\begin{itemize}
	\item $\per(g)$ is not a power of $p$, or
	\item $\per(g)$ is a power of $p$ and $\per(g) \geq \sqrt{n\log(1/\varepsilon)},$ or
	\item $B(h) = 0.$
	\end{itemize}
	
	So from now, we assume that $\per(g)$ is a power of $p$ upper-bounded by $\sqrt{n\log(1/\varepsilon)}$ and that $B(h)\geq 1.$ In this case, Lemma~\ref{lem:preg-f-weak} shows that $\pdeg(f) \geq \widetilde{\Omega}(\per(g)).$ On the other hand, since $B(h) \leq n$ and $\varepsilon\geq 2^{-n},$ the lower bound we need to show is $\widetilde{\Omega}(\per(g) + \sqrt{B(h)\log(1/\varepsilon)}+\log(1/\varepsilon)).$ By Lemma~\ref{lem:pdeg-h-improv}, it suffices to show a lower bound of $\widetilde{\Omega}(\per(g)+\pdeg_\varepsilon(h)).$
	
	The analysis splits into two simple cases based on the relative magnitudes of $\per(g)$ and $\pdeg_\varepsilon(h)$.
	
	Assume first that $\pdeg_{\varepsilon}(h) \leq 4\cdot\per(g)$. In this case, we are trivially done, because we already have $\pdeg(f) = \widetilde{\Omega}(\per(g))$, which is $\widetilde{\Omega}(\pdeg(g) + \pdeg_\varepsilon(h))$ as a result of our assumption.
	
	Now assume that $\pdeg_\varepsilon(h) > 4\cdot \per(g).$ We know that $f = g\oplus h$ and hence $h = f\oplus g.$ Hence, we have
	\begin{align*}
	\pdeg_\varepsilon(h) \leq 2(\pdeg_{\varepsilon/2}(f) + \pdeg_{\varepsilon/2}(g)) \leq  O(\pdeg_{\varepsilon}(f)) + 2\cdot \per(g),
	\end{align*}
	where the first inequality is a consequence of Fact~\ref{fac:pdeg} item 2 and the second is a consequence of Fact~\ref{fac:pdeg} item 1 and Theorem~\ref{thm:main-ubd}. The above yields
	\[
	\pdeg_{\varepsilon}(f) = \Omega( (\pdeg_\varepsilon(h) - 2\cdot \per(g)) ) = \Omega(\pdeg_\varepsilon(h))  = \Omega(\per(g) + \pdeg_\varepsilon(h)).
	\]
	This finishes the proof.
	\end{proof}		
	
	\subparagraph*{Acknowledgements.} We are grateful to Siddharth Bhandari and Tulasimohan Molli for discussions at the start of this project. We are also grateful to the anonymous reviewers for FSTTCS 2019 for their many comments that helped improve this paper (in particular, one of the reviewers pointed out the problem of getting tight lower bounds as a function of the error parameter, which allowed us to strengthen our results).
	
	\bibliographystyle{alpha}
	\bibliography{symm-references}

\begin{thebibliography}{BHMS18}

\bibitem[AW15]{alman}
Josh Alman and Ryan Williams.
\newblock Probabilistic polynomials and hamming nearest neighbors.
\newblock In {\em 2015 IEEE 56th Annual Symposium on Foundations of Computer
  Science (FOCS)}, pages 136--150. IEEE, 2015.

\bibitem[Bei93]{Beigel-survey}
Richard Beigel.
\newblock The polynomial method in circuit complexity.
\newblock {\em [1993] Proceedings of the Eigth Annual Structure in Complexity
  Theory Conference}, pages 82--95, 1993.

\bibitem[BHMS18]{BHMS}
Siddharth Bhandari, Prahladh Harsha, Tulasimohan Molli, and Srikanth
  Srinivasan.
\newblock {On the Probabilistic Degree of OR over the Reals}.
\newblock In Sumit Ganguly and Paritosh Pandya, editors, {\em 38th IARCS Annual
  Conference on Foundations of Software Technology and Theoretical Computer
  Science (FSTTCS 2018)}, volume 122 of {\em Leibniz International Proceedings
  in Informatics (LIPIcs)}, pages 5:1--5:12, Dagstuhl, Germany, 2018. Schloss
  Dagstuhl--Leibniz-Zentrum fuer Informatik.

\bibitem[Bra10]{Brav}
Mark Braverman.
\newblock Polylogarithmic independence fools \emph{AC}\({}^{\mbox{0}}\)
  circuits.
\newblock {\em J. {ACM}}, 57(5), 2010.

\bibitem[BRS91]{BRS}
Richard Beigel, Nick Reingold, and Daniel~A. Spielman.
\newblock The perceptron strikes back.
\newblock In {\em Proceedings of the Sixth Annual Structure in Complexity
  Theory Conference, Chicago, Illinois, USA, June 30 - July 3, 1991}, pages
  286--291, 1991.

\bibitem[BW87]{BW}
Bettina Brustmann and Ingo Wegener.
\newblock The complexity of symmetric functions in bounded-depth circuits.
\newblock {\em Inf. Process. Lett.}, 25(4):217--219, 1987.

\bibitem[DP09]{dubhashi_panconesi_2009}
Devdatt~P Dubhashi and Alessandro Panconesi.
\newblock {\em Concentration of measure for the analysis of randomized
  algorithms}.
\newblock Cambridge University Press, 2009.

\bibitem[FKPS85]{FKPS}
Ronald Fagin, Maria~M. Klawe, Nicholas Pippenger, and Larry~J. Stockmeyer.
\newblock Bounded-depth, polynomial-size circuits for symmetric functions.
\newblock {\em Theor. Comput. Sci.}, 36:239--250, 1985.

\bibitem[HS16]{HS}
Prahladh Harsha and Srikanth Srinivasan.
\newblock {On Polynomial Approximations to $AC^0$}.
\newblock In Klaus Jansen, Claire Mathieu, Jos{\'e} D.~P. Rolim, and Chris
  Umans, editors, {\em Approximation, Randomization, and Combinatorial
  Optimization. Algorithms and Techniques (APPROX/RANDOM 2016)}, volume~60 of
  {\em Leibniz International Proceedings in Informatics (LIPIcs)}, pages
  32:1--32:14, Dagstuhl, Germany, 2016. Schloss Dagstuhl--Leibniz-Zentrum fuer
  Informatik.

\bibitem[JJJ80]{johnson-stringlemma}
David~Lawrence Johnson, David~Leroy Johnson, and SS~Johnson.
\newblock {\em Topics in the theory of group presentations}, volume~42.
\newblock Cambridge University Press, 1980.

\bibitem[LN90]{LN}
Nathan Linial and Noam Nisan.
\newblock Approximate inclusion-exclusion.
\newblock {\em Combinatorica}, 10(4):349--365, 1990.

\bibitem[Lu01]{Lu}
Chi-Jen Lu.
\newblock An exact characterization of symmetric functions in
  $\mathrm{qAC^0[2]}$.
\newblock {\em Theoretical Computer Science}, 261(2):297--303, 2001.

\bibitem[Luc78]{Lucas-thm}
Edouard Lucas.
\newblock Th{\'e}orie des fonctions num{\'e}riques simplement p{\'e}riodiques.
\newblock {\em American Journal of Mathematics}, pages 289--321, 1878.

\bibitem[MNV16]{MNV}
Raghu Meka, Oanh Nguyen, and Van Vu.
\newblock Anti-concentration for polynomials of independent random variables.
\newblock {\em Theory of Computing}, 12(1):1--17, 2016.

\bibitem[O'D14]{ODonnell}
Ryan O'Donnell.
\newblock {\em Analysis of Boolean Functions}.
\newblock Cambridge University Press, New York, NY, USA, 2014.

\bibitem[Pat92]{paturi}
Ramamohan Paturi.
\newblock On the degree of polynomials that approximate symmetric boolean
  functions (preliminary version).
\newblock In {\em Proceedings of the twenty-fourth annual ACM symposium on
  Theory of computing}, pages 468--474. ACM, 1992.

\bibitem[Raz87]{Razbo}
Alexander~A. Razborov.
\newblock {L}ower bounds on the size of bounded depth circuits over a complete
  basis with logical addition.
\newblock {\em Mathematicheskie Zametki}, 41(4):598--607, 1987.
\newblock (English translation in {\em Mathematical Notes of the Academy of
  Sciences of the USSR}, 41(4):333--338, 1987).

\bibitem[Smo87a]{Smolensky87}
Roman Smolensky.
\newblock Algebraic methods in the theory of lower bounds for boolean circuit
  complexity.
\newblock In {\em Proceedings of the nineteenth annual ACM symposium on Theory
  of computing}, pages 77--82. ACM, 1987.

\bibitem[Smo87b]{Smolen}
Roman Smolensky.
\newblock Algebraic methods in the theory of lower bounds for boolean circuit
  complexity.
\newblock In {\em Proceedings of the 19th Annual ACM Symposium on Theory of
  Computing}, pages 77--82, 1987.

\bibitem[Smo93a]{Smolensky93}
Roman Smolensky.
\newblock On representations by low-degree polynomials.
\newblock In {\em Proceedings of 1993 IEEE 34th Annual Foundations of Computer
  Science}, pages 130--138. IEEE, 1993.

\bibitem[Smo93b]{Smolen93}
Roman Smolensky.
\newblock On representations by low-degree polynomials.
\newblock In {\em FOCS}, pages 130--138, 1993.

\bibitem[Tar93]{Tarui}
Jun Tarui.
\newblock Probablistic polynomials, {$\AC^0$} functions, and the
  polynomial-time hierarchy.
\newblock {\em Theoretical Computer Science}, 113(1):167--183, 1993.

\bibitem[Wil14]{Williams-survey}
Richard~Ryan Williams.
\newblock The polynomial method in circuit complexity applied to algorithm
  design (invited talk).
\newblock In {\em 34th International Conference on Foundation of Software
  Technology and Theoretical Computer Science (FSTTCS 2014)}. Schloss
  Dagstuhl-Leibniz-Zentrum fuer Informatik, 2014.

\bibitem[ZBT93]{ZhangBarringtonTarui}
Zhi-Li Zhang, David A~Mix Barrington, and Jun Tarui.
\newblock Computing symmetric functions with and/or circuits and a single
  majority gate.
\newblock In {\em Annual Symposium on Theoretical Aspects of Computer Science},
  pages 535--544. Springer, 1993.

\end{thebibliography}
	
\end{document}